\documentclass[sigconf,nonacm]{acmart}

\makeatletter
\def\@ACM@checkaffil{
    \if@ACM@instpresent\else
    \ClassWarningNoLine{\@classname}{No institution present for an affiliation}%
    \fi
    \if@ACM@citypresent\else
    \ClassWarningNoLine{\@classname}{No city present for an affiliation}%
    \fi
    \if@ACM@countrypresent\else
        \ClassWarningNoLine{\@classname}{No country present for an affiliation}%
    \fi
}
\makeatother

\usepackage{algorithm,amsmath,amsthm,algorithmic}
\usepackage{float,graphicx,subcaption}
\usepackage{mathtools,array,gensymb}
\usepackage{tabularx,pifont}
\usepackage{multicol,multirow}
\usepackage{framed}
\usepackage{xcolor,url}
\usepackage{etoolbox}

\newtheorem{proposition}{Proposition}[section]

\newrobustcmd{\bb}{\bfseries}

\begin{document}

\title{Secret-Free Device Pairing in the mmWave Band}

\author{Ziqi Xu\textsuperscript{\textdagger}\hspace{0.5cm}
Jingcheng Li\textsuperscript{\textdagger}\hspace{0.5cm} 
Yanjun Pan\textsuperscript{\textdaggerdbl}\hspace{0.5cm}
Ming Li\textsuperscript{\textdagger}\hspace{0.5cm}
Loukas Lazos\textsuperscript{\textdagger}}
\affiliation{\institution{\textdagger~ Department of Electrical and Computer Engineering, University of Arizona}}
\affiliation{\institution{\textdaggerdbl~ Department of Computer Science and Computer Engineering, University of Arkansas}}
\affiliation{Email: \textdagger~\{zxu1969, jli2972, lim, llazos\}@arizona.edu, \textdaggerdbl~ yanjunp@uark.edu}

\begin{abstract}

Many Next Generation (NextG)  applications feature devices that are capable of communicating and sensing in the Millimeter-Wave (mmWave) bands. Trust establishment is an important first step to bootstrap secure mmWave communication links,  which is challenging due to the lack of prior secrets and the fact that traditional cryptographic authentication methods cannot bind digital trust with physical properties. Previously, context-based device pairing approaches were proposed to extract shared secrets from common context, using various sensing modalities. However, they suffer from  various limitations in practicality and security. 

In this work, we propose the first secret-free device pairing scheme in the mmWave band that explores the unique physical-layer properties of mmWave communications. Our basic idea is to let Alice and Bob derive common randomness by sampling physical activity in the surrounding environment that disturbs their wireless channel. They construct  reliable fingerprints of the activity by extracting event timing information from the channel state. We further propose an uncoordinated path hopping mechanism to resolve the challenges of beam alignment for activity sensing without prior trust. A key novelty of our protocol is that it remains secure against  both  co-located passive adversaries and active Man-in-the-Middle attacks, which is not possible with existing context-based pairing approaches.  
We implement our protocol in a 28GHz mmWave testbed, and experimentally evaluate its  security in realistic indoor environments. Results show that our protocol can effectively thwart several different types of adversaries. 

\end{abstract}

\pagestyle{plain}
\settopmatter{printfolios=true} 

\maketitle

\section{Introduction}

The increasing demand for higher data rates has pushed wireless technologies to millimeter wave (mmWave) bands where bandwidth is abundant.  The physics of signal propagation at such short wavelengths offer some security advantages compared to typical microwave communications due to the substantial signal attenuation (which is proportional to the center frequency) and susceptibility to blockage \cite{7999294}. For instance, eavesdropping becomes more challenging as mmWave transmissions are highly-directional and do not travel through physical barriers. Nevertheless, common security threats such as eavesdropping, message injection, message modification, impersonation, and Man-in-the-Middle (MitM) attacks  are still feasible \cite{balakrishnan2019modeling, steinmetzer2018beam} if wireless transmissions are left unprotected.

The widely adopted approach for protecting wireless links is to incorporate cryptographic solutions relying on asymmetric and symmetric cryptosystems \cite{stallings2012computer}. However, such solutions require extensive key management for establishing and maintaining trust, which can pose scalability challenges. Moreover, digital keys created from pseudo-random generators do not bind the devices that use them with any surrounding physical property. For instance, two devices that use a symmetric key to encrypt communications over a public channel, cannot verify any other physical property such as location. Extending trust to the physical world has become prevalent with the advent of  Next Generation (NextG) communications, where many applications seamlessly integrate the cyber and physical worlds (e.g., industrial automation, connected healthcare, connected autonomous vehicles).

A complementary approach to pre-loading pseudo-randomly generated secrets is to generate secrets from common sources of randomness that exhibit high spatial and temporal dynamics \cite{xu2021pof}. The main idea is that two co-located devices can independently sample a high-entropy physical property to derive a common key. The high spatial correlation of the samples at the locations of the legitimate devices compared to the low correlation at the adversary's location offers the security advantage in the key generation process. The so-called context-based key generation methods have notable advantages such as reduced user efforts, perpetual key renewal, and high level of security to remote adversaries. Commonly explored physical modalities include light \cite{miettinen2014context}, sound \cite{xu2021inaudiblekey, miettinen2014context, schurmann2011secure}, wireless channel reciprocity \cite{liu2013fast, mathur2008radio}, ambient RF signals \cite{mathur2011proximate, jana2009effectiveness}, human operation \cite{li2020t2pair, zhang2018tap}, as well as multi-modal approaches \cite{han2018you, farrukh2022one}.

\subsection{Motivation and Challenges}

Context-based secure pairing methods present numerous challenges. First, many modalities (e.g., light, sound) suffer from low entropy, requiring long sampling times to agree on common bits. They also require all devices to possess a common sensing capability, which may not be available in IoT applications. More importantly, they require physical access control to guarantee security. The threat model assumes that the adversary has no physical access to the environment where measurements take place. Moreover, they are susceptible to active attacks where an adversary alters the physical environment to make key generation predictable. RF modalities exhibit high entropy in rich-multipath environments, which are common in sub-6 GHz communications where the majority of prior works have been operating \cite{mathur2011proximate, liu2013fast, liu2012exploiting}. However, mmWave communications follow geometric channel models \cite{akdeniz2014millimeter} with very few reflections contributing to the channel uncertainty. Moreover, RF based methods have been shown to be vulnerable to active attacks where an adversary, Mallory, opportunistically injects signals to make bit extraction predictable \cite{eberz2012practical}. 

\begin{figure}[t]
\centering
\includegraphics[width=0.7\columnwidth]{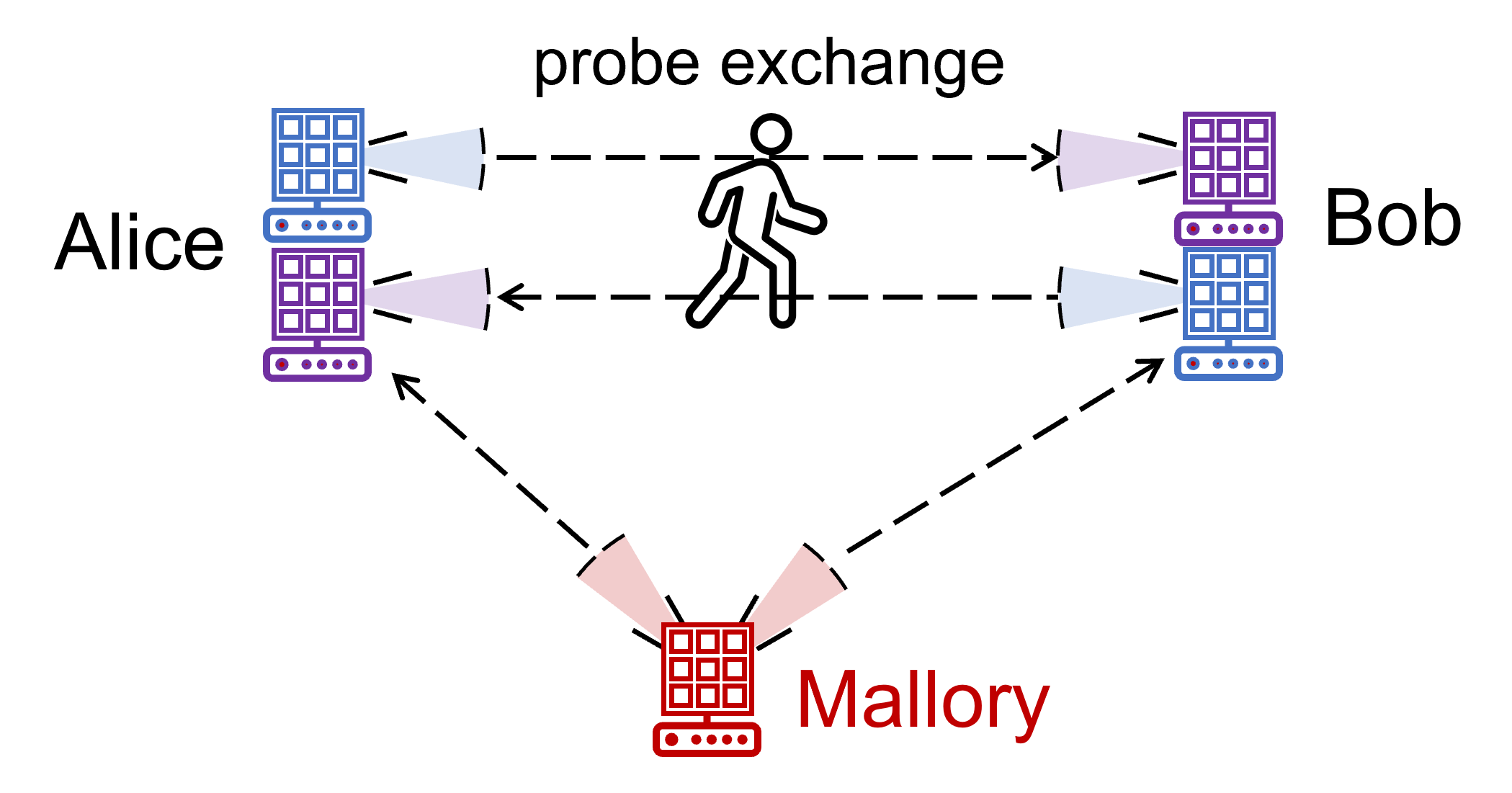}
\vspace{-0.1in}
\caption{Alice and Bob use mmWave probes to detect physical activity in the line-of-sight path and agree on a common key. Mallory  cannot sample the same activity.}
\label{fig:system}
\vspace{-0.15in}
\end{figure}

To counter the limitations of prior works and to transition device pairing protocols to  NextG  networks, we explore the unique PHY-layer properties of mmWave communications to develop a novel secret-free device pairing protocol. Specifically, Alice and Bob independently derive common randomness by sampling physical activity that disturbs their wireless channel in the presence of Mallory. Wireless paths are interrupted by artificial or naturally-occurring motion (e.g., human motion) within the environment at random times. This interruption is sensed by measuring the amplitude of channel state information (CSI) over the distinct paths that connect Alice and Bob. An example is shown in Figure~\ref{fig:system}. The antennas of Alice and Bob are aligned over the line-of-sight (LOS) path. The LoS path is interrupted by physical activity. Even if Mallory is in the same environment, she observes a different path and therefore cannot measure the same CSI variations. Mallory could observe the physical activity using other modalities such as vision; however, it is exceedingly difficult to estimate the impact of motion on the wireless channel at a very fine time scale, thus leaving the  pairing process secure.

There are several key technical challenges to design such a protocol that senses ambient activity for device pairing in the mmWave band. First, due to the small wavelength of mmWave signals, the CSI (and its amplitude)  in the mmWave band is very sensitive to small disturbances to the channel, and channel reciprocity may not hold. Directly adopting measured CSI samples  for key agreement will introduce a high bit mismatch rate.  Second,  since mmWave communication is highly directional,   devices must have their beams  aligned on the direction(s) where activity exists to be able to sense it. However, neither prior knowledge nor prior trust is available to support such beam alignment in a secure manner. Third, the protocol must be resistant against a wide range of threats, not only co-present passive adversaries  but also active attacks such as Man-in-the-Middle. This should be achieved without prior trust. 

\subsection{Main Contributions}
In this paper, we propose a novel context-based key pairing protocol for devices that operate in the mmWave band, with merely in-band RF communication. The basic protocol is called JellyBean, which exploits the pencil-like beams used in mmWave communications to detect disturbances (channel blockage) to the wireless channel due to physical motion from natural or artificial activity. To solve the first challenge, we extract the timing of activities from the CSI magnitude to generate a  fingerprint of the ambient environment, which is used as a common source of randomness between Alice and Bob. A passive adversary is unable to infer the secret if observing the same activities from a different vantage point due to the high directionality of mmWave transmissions. 

Our basic pairing protocol is still vulnerable to a passive adversary that is located in close proximity to either Alice or Bob and aligned in the sampled path or an  active adversary that launches beam-stealing attacks \cite{steinmetzer2018beam}.  Thus, we further introduce Jellybean+ which incorporates 
an uncoordinated path-hopping (UPH) mechanism.
UPH takes advantage of the electronic antenna steering capabilities at mmWave to randomize the paths sampled by Alice and Bob in a similar way as channel-hopping randomizes the communication frequency. Without knowledge of which paths are sampled at any point in time, a co-located adversary cannot construct fingerprints similar to Alice and Bob. A notable difference from frequency hoping is that in UPH no common hopping sequence is necessary for activity sensing.

We experimentally evaluate the performance and security of our protocols on a 28GHz mmWave testbed in typical indoor environments. Our experiments demonstrate that the the bit mismatch rate between Mallory and Alice or Bob far exceeds that of the legitimate devices, allowing for secure pairing in reasonable time. We  also show that the key derived from physical activities satisfies the NIST randomness tests \cite{bassham2010sp}.  To further demonstrate the security of our protocols, we analyze both passive attacks involving eavesdropping as well as recording of physical activities using cameras.  

 The key novelty of Jellybean compared to the state-of-the-art context-based pairing protocols is that it remains secure even if Mallory (a) Mallory is co-located in the same environment and can observe the ongoing physical activity and (b) launches an  active attack. In fact, Jellybean+ remains secure even if Mallory is at the same location as Alice or Bob.

\section{System and Threat Model}
\label{sec:system-model}

\subsection{System Model}

The system consists of an access point $A$ and a legitimate device $D$ located in a typical indoor environment. Both devices operate in the mmWave band and are equipped with electronically-steerable directional antennas. Their antennas can be aligned using the Sector Level Sweep (SLS) protocol of IEEE 802.11ad \cite{IEEE:802.11ad}, although any suitable beam alignment protocol can be used. The devices aim at securely pairing and establish a common secret $k_A = k_D =k$. The devices do not share any prior secrets.

\begin{figure}[t]
\centering
\setlength{\tabcolsep}{-3pt}
\begin{tabular}{cc}
  \includegraphics[width=0.52\columnwidth]{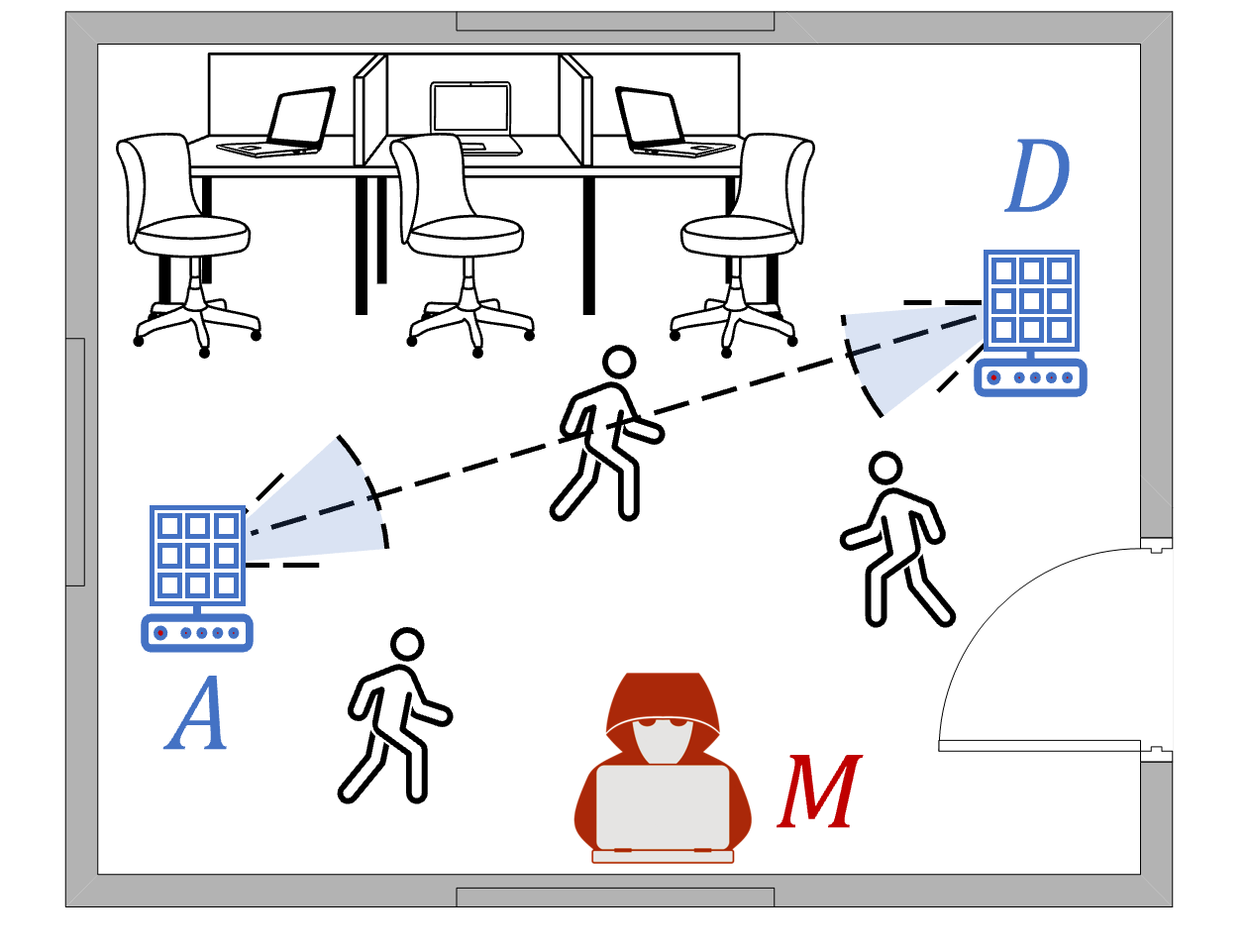}&
  \includegraphics[width=0.52\columnwidth]{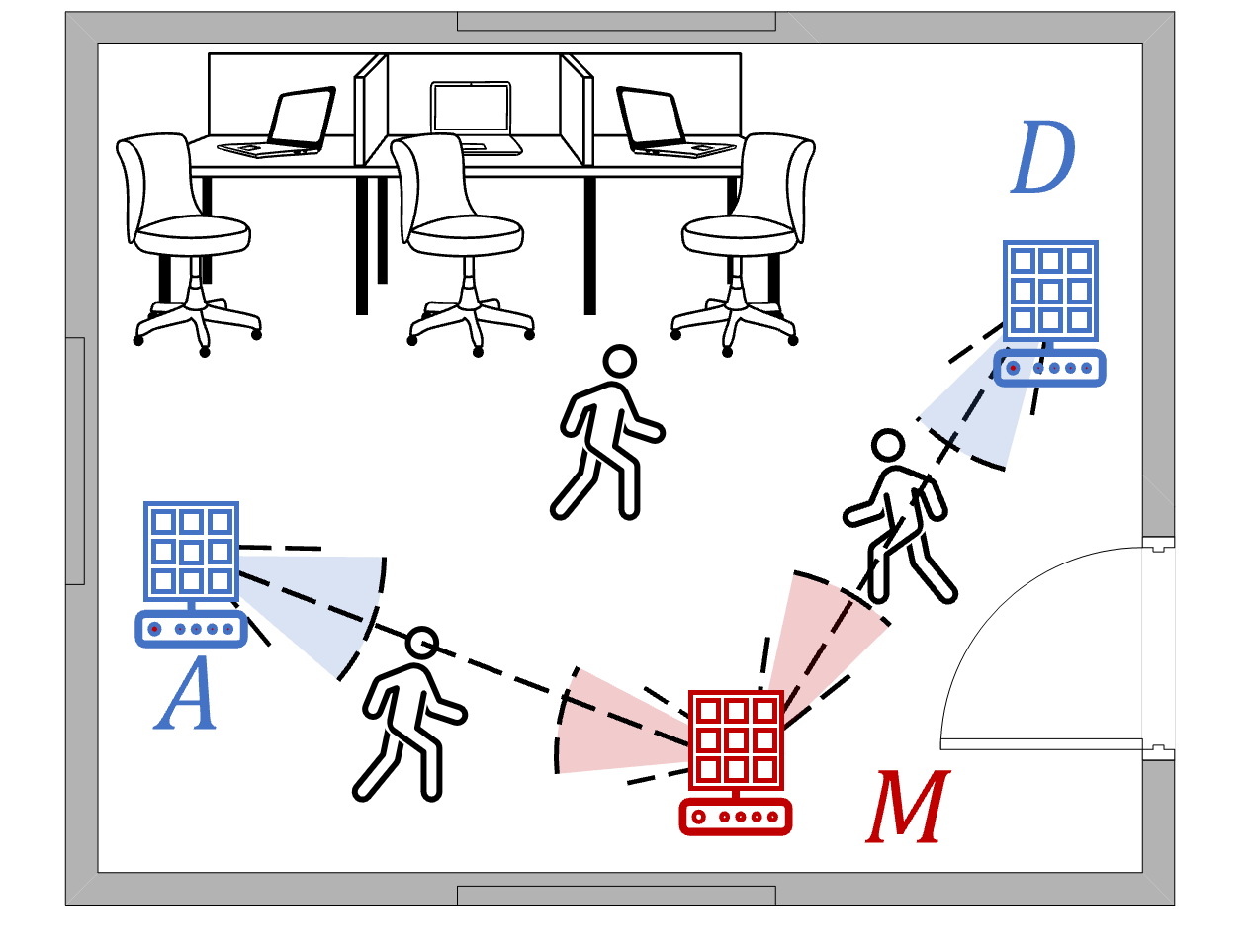} \\
  (a) Passive attack & (b) Active beam-stealing attack
\end{tabular}
\vspace{-0.1in}
\caption{Threat model (a) Mallory passively observes the same environment as $A$ and $D$, (b) Mallory launches an active beam stealing attack.}
\label{fig:threatmodel}
\vspace{-0.2in}
\end{figure}

\subsection{Threat Model}

The goal of the adversary is to learn the common secret $k$ derived by $A$ and $D$ during the pairing process. The adversary is equipped with an electronically steerable antenna that can be pointed to any direction. Most prior context-based pairing methods assume physical access control \cite{han2018you, farrukh2022one, schurmann2011secure, miettinen2014context}. The adversary is assumed to be unable to observe the environment sampled by $A$ and $D.$ For instance, the legitimate devices and the adversary are separated by a physical barrier such as a wall that makes physical properties such as light, sound, and human activity unobservable. Our threat model breaks this rather artificial assumption and allows the adversary to be in the same environment as $A$ and $D$, as shown in Figure~\ref{fig:threatmodel}(a).

The adversary has full knowledge of the context-based pairing protocol executed by the legitimate devices. However, she does not control the physical activities taking place in the environment as this is easily visually detected. Mallory could be either passively trying to derive the secret $k$ or launch an active attack.

{\bf Passive attacks:} A passive Mallory can steer her antenna to $A$ and $D$ respectively, to eavesdrop on any messages exchanged between the two parties. Moreover, Mallory can steer her antenna to any other desired direction including the location where physical activities may take place. Mallory can position herself at a location of her own choosing. This assumption is fairly strong, as the presence of Mallory in private spaces can be easily detected. However, it may hold for public spaces. Moreover, it provides the worst-case analysis on the protocol  security. Finally, Mallory may have access to a separate modality such as vision (camera) that allows her to observe the physical activities occurring in the environment. The resolution is limited by the modality used (e.g., frame rate). 
  
{\bf Active Attacks.} Mallory can inject her own signals to compromise the device pairing process. We assume that the purpose of signal injection is the compromise of the secret derived between $A$ and $D$ \cite{pan2021man} and not a denial-of-service (DoS) attack. The latter is fairly easy to launch against device pairing protocols using interference attacks such as selective-jamming.  Moreover, Mallory can launch a beam-stealing attack  that allows the implementation of a MitM adversary in directional mmWave communications \cite{steinmetzer2018beam}. The main idea of a beam stealing is shown in Figure~\ref{fig:threatmodel}(b). Mallory manipulates the beam alignment process between $A$ and $D$ by forging sector sweep frames with high RSS. $A$ and $D$ can only communicate through Mallory as their antennas are aligned to $M$. Mallory can then launch a MitM attack against the pairing process between $A$ and $D.$ We note that although several methods have been proposed to add authentication to the beam alignment process \cite{steinmetzer2018authenticating, balakrishnan2019physical, wang2021exploiting, wang2020machine}, they require prior trust to be established, or training for the purpose of fingerprinting and thus are not applicable in the absence of trust.

\section{Protocol Design}

We design a context-based secure pairing protocol suitable for mmWave communications. Security is drawn from the random physical motions in the environment (e.g., human activity).  The motions could occur naturally in high-traffic scenarios, or could be artificially induced by the user. We emphasize that although physical activity has been used as a source of randomness in prior works \cite{li2020t2pair, zhang2018tap}, it affects mmWave transmissions in a different way than microwave transmissions. In microwave frequencies, motion is observed primarily via multipath and scattering whereas in mmWave, it is primarily observed as blockage, which is path-specific. This path specificity provides an opportunity to build stronger security compared to microwave-based methods, that are not to be vulnerable to non-LoS sensing \cite{zhu2020tu}.

The legitimate devices $A$ and $D$ align their antennas over a feasible communication path using a beam alignment protocol such as the sector level sweep process in 802.11ad \cite{IEEE:802.11ad}. The random motions in the environment disrupt the $A-D$ channel in an unpredictable manner in terms of timing and duration. These disruptions are unique to the path between $A-D$ due to the high antenna directionality of mmWave communications and prevent Mallory from recording similar channel disruptions by observing the same physical activity from a different vantage point.

Some vantage points may allow the observation of similar disruptions as the $A-D$ channel. For instance, when $M$ is aligned with the $A-D$ LoS channel. To secure against these special locations, we add diversity in the paths sampled for activity. Rather than relying on a single path, $A$ and $D$ randomly hop to other available paths (e.g., through a reflection) in an uncoordinated fashion, thus randomizing the direction over which activity is observed. Even if the adversary can measure channel activity on all paths (e.g., by fast switching or deploying multiple antennas), she does not know the antenna directions for $A$ and $D$, thus being unable to infer which activity is used in key generation. Our uncoordinated path hopping method resembles frequency hopping to avoid eavesdropping and jamming \cite{strasser2008jamming}, but it occurs without a commonly shared hopping sequence and in the space domain. 

\subsection{Fingerprinting Physical Activity}
\label{sec:AFG}

The core security primitively used in our secure key extraction process is the unpredictability of how physical activity impacts a mmWave path at a fine time scale. To fingerprint physical activity, the two legitimate devices $A$ and $D$ align their antennas over a feasible communication path (e.g., the LoS path), as shown in Figure~\ref{fig:threatmodel}(a). The two devices exchange probe messages to measure the CSI amplitude independently. These measurements are processed to extract the start and end time of the physical activity, which is used to independently derive common bits.  We emphasize, that contrary to previous CSI-based methods that exploit channel reciprocity \cite{mathur2008radio}, we do not use the CSI amplitude directly for key extraction. This is because the CSI has low entropy in mmWave frequencies due to its geometric nature when the path is not disrupted and channel reciprocity breaks under motion on the sensed path. 

The main challenge lies in independently processing the CSI samples to achieve a high bit agreement rate while maintaining the entropy provided by the physical motion. We meet this challenge by applying the processing steps shown in Figure~\ref{fig:AFG}. First, the
CSI samples are denoised and smoothed to filter inherent channel fluctuations introduced by the fine nature of CSI measurement in mmWave. Then, we apply threshold detection to convert CSI sample sequences to a bitstream indicating the presence or absence of activity. Finally, we encode the bitstream to increase the entropy of the derived common secret. We describe each step in detail for device $A$. The same process is applied at $D$. 

\begin{figure*}[t]
\centering
\includegraphics[width=2\columnwidth]{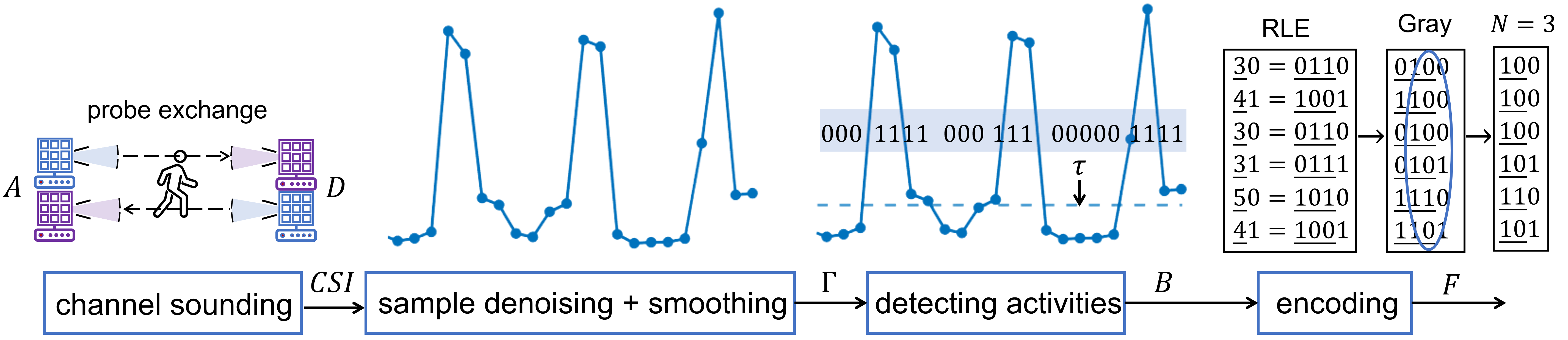}
\caption{Workflow of the activity fingerprinting process.}
\label{fig:AFG}
\vspace{-0.1in}
\end{figure*}

\textbf{1. Channel sounding.} To sense physical activity, devices $A$ and $D$ exchange probes used for channel sounding \cite{rappaport2010wireless}. The exchange of two probes completes one probing round. At the end of $N$ rounds, the CSI samples are organized into sets $CSI_A$ and $CSI_D$, respectively. We focus on set $CSI_A$ and drop the $A$ subscript for ease of notation.  The CSI set is denoted by $CSI =\{c_1,c_2,\ldots c_K\}$ and each $c_i$ is the set of samples collected per subcarrier. The probes are relatively short, containing a preamble and the unauthenticated device ID.

\textbf{2. Sample denoising and smoothing.} To reduce the noise in CSI measurements, the wavelet transform is applied on each $c_i$ \cite{yu2018passive}. As channels may be affected differently at various subcarrier frequencies, we compute the cross-correlation between the CSI amplitude samples at each subcarrier and the center subcarrier $K/2$. For two subcarriers $i$ and $j$ with CSI sample sequences $c_i$ and $c_j$ the cross-correlation is given by
\begin{align}   
 \rho (i, j)=\frac{1}{\alpha N-1}\sum\limits_{k=1}^{\alpha N}(\frac{c_i(k)-\mu_{c_i}}{\sigma_{c_i}})(\frac{c_j(k)-\mu_{c_i}}{\sigma_{c_j}})
\label{eq:CC}
\end{align}
where $\mu_{c_i}$ and $\sigma_{c_i}$ are the mean and variance of set $c_i$ and $\alpha$ is the number of CSI samples extracted per probe and per subcarrier. The sample sets $c_1, c_2,\ldots c_K$ are ranked according to their cross-correlation with the $K/2^{th}$ subcarrier and only the top half is retained. The new set of samples is denoted by $CSI =\{c_1,c_2,\ldots c_{K/2}\}$ where the set indices correspond to the ranked list of subcarriers.  

To further alleviate the high sensitivity of the CSI amplitude measurements in the mmWave channel, we apply a $W_{m}$-point moving variance operation to the CSI amplitude along each subcarrier set $c_i$. Given a set of samples $c_i=\{c_i(1),c_i(2),\cdots\,c_i((\alpha N)\}$, a $W_{m}$-point moving variance for a sample $c_i(j)$ is computed by
\[
MV(c_i(j))=\frac{1}{W_m-1}\sum\limits_{c(m)\in X_j}(c(m)-\mu_j)^2,
\]
where $X_j=\{c_i(j-\lfloor \frac{W_m-1}{2} \rfloor),\cdots,c_i(j),\cdots,c_i(j+\lfloor \frac{W_m-1}{2} \rfloor)\}$ and $\mu_j$ denotes the mean of $X_j$. We opt to use CSI amplitude variance to capture changes in the CSI amplitude because they are indicative of channel blockage. The moving variance operator is applied to smooth the rapid fluctuations. In the next step, the $W_m$-point moving variance is averaged over the $K/2$ highest correlated subcarriers to produce a single time series of CSI amplitude, denoted as $\mathcal{T}=\{c(1), c(2),\ldots, c(\alpha N)\}.$ Finally, to cope with the lack of perfect channel reciprocity and the sensitivity of the mmWave channel to small geometry changes \cite{zhang2015coverage}, an averaging window of $W_{s}$ is applied across $\mathcal{T}$ to downsample it by a factor of $W_s$. The final CSI sequence is denoted by $\Gamma$ and has a size of $\lfloor \alpha N/ W_s \rfloor$ samples.

\textbf{3. Detecting Activity.} Physical activity on a time series $\Gamma$ is inferred by using threshold-based detection \cite{rousseeuw1993alternatives}. The process generates a bit sequence $B$ by comparing every sample in $\Gamma$ with a detection threshold $\tau$:
\begin{align*}
& b(i) =\left\{
  \begin{aligned}
    1, &\quad \gamma(i) \geq \tau, \\
    0, &\quad \text{otherwise}. \\
  \end{aligned}
\right.
\end{align*}
The threshold $\tau$ is defined as the average over the $\ell$ lowest CSI amplitude samples in $\Gamma$, where $\ell$ corresponds to a short time duration (e.g., one second). Given the time scale of physical motion, this duration will correspond to a period of inactivity.  

\textbf{4. Encoding.} Due to the high sampling rate relative to the speed of any physical motion, bit sequence $B$ consists of many consecutive zero bits (inactivity) followed by consecutive one bits (activity), thus having relatively low entropy. The random part of the physical motion sensed by the probes lies in the start time and duration of the channel disruption. To encode that information alone, we apply run-length encoding (RLE) to transform sequence $B$ into blocks. In RLE, each block of consecutive bits (zeros or ones) is encoded to a counter representing the block length (duration), followed by the block's bit value (activity or inactivity). An example of the application of RLE to an example $B$ is given below (the underlined bits represent the block counter). 
\[
\begin{tabular}{rcccccc}
$B$ =
&000000&1111&000&111&00000&1111\\
RLE[$B$] =
&\underline{6}0 &\underline{4}1 &\underline{3}0 &\underline{3}1 &\underline{5}0 &\underline{4}1\\
 =
&\colorbox{lightgray}{\underline{110}0}
&\colorbox{lightgray}{\underline{100}1} 
&\colorbox{lightgray}{\underline{011}0} 
&\colorbox{lightgray}{\underline{011}1}
&\colorbox{lightgray}{\underline{101}0}
&\colorbox{lightgray}{\underline{100}1}\\
\end{tabular}
\]
Additionally, we apply Gray encoding to reduce the mismatch between the sequences generated by $A$ and $D$ due to slight mismatches in the block lengths. The number of bits used in Gray code is determined by the longest consecutive bits counted by RLE. In the example, the longest consecutive bit is 6. So we use a 3-bit Gray code in each block. In Gray encoding, two successive counters differ by exactly one bit. The Gray-encoded sequence for the example above is represented by $S$:
\[
\begin{tabular}{ccccccc}
$S$ =
&\colorbox{lightgray}{\underline{101}0}
&\colorbox{lightgray}{\underline{110}1}
&\colorbox{lightgray}{\underline{010}0}
&\colorbox{lightgray}{\underline{010}1}
&\colorbox{lightgray}{\underline{111}0}
&\colorbox{lightgray}{\underline{110}1}
\end{tabular}
\]
The final step is to truncate each block to the same length, so that events contribute equally to the final sequence. We achieve this by retaining $N$ least significant bits (LSBs) from each gray-coded block. We retain LSBs because they are harder to predict and thus increase the entropy of the sequence. The final sequence upon the application of truncation of $N=3$ in the running example is 
\[
F=010~101~100~101~110~101.
\]
Sequence $F$ corresponds to the final bit sequence that will be used as an activity fingerprint in the agreement to a common key. At the end of the independent CSI measurement processing, each device $A$ and $D$ has computed $F_A$ and $F_D$, respectively.

\subsection{The Jellybean Pairing Protocol}

We integrate the activity fingerprinting process into a secure device pairing protocol called {\em Jellybean.} The Jellybean protocol allows  $A$ and $D$ establish a common secret $k$ and prevents Mallory from learning $k$. Moreover, it provides authentication through presence. If $M$ were to execute the same protocol with $A$ (either asynchronously or in parallel), she would establish a different secret. The protocol consists of the activity fingerprinting (AF) phase followed by the key agreement phase, which is a standard fuzzy-commitment based key agreement \cite{juels1999fuzzy}. The two phases are illustrated in Figure~\ref{fig:diagram} and described below.

\noindent
\textbf{\textit{Initialization}}
\begin{enumerate}

  \item $A$ and $D$ align their antennas using a beam alignment protocol \cite{IEEE:802.11ad}.
   \item $A$ and $D$ agree on parameters $W_m, W_s,N$, and  the Reed-Solomon code $RS$. The parameters  could also be preset. 
\end{enumerate}

\noindent
\textbf{\textit{Activity Fingerprinting ($AF$)}} 
\begin{enumerate}
\setcounter{enumi}{4}
  \item $D$ and $A$ execute probing rounds and collect $CSI_D$ and $CSI_A$, respectively.
  
  \item $D$ and $A$ generate fingerprints $F_D$ and $F_A$, respectively, using the activity fingerprinting process described in Section \ref{sec:AFG}.
\end{enumerate} 

\noindent
\textbf{\textit{Key Agreement}}
\begin{enumerate}
\setcounter{enumi}{6}
  \item $A$ generates a random secret key $k$ and creates  commitment $\delta_A=F_A \ominus ENC_{RS}(k)$, where $\ominus$ denotes subtraction in a Galois Field of $\mathbb{F}_q^n$, and $ENC_{RS}(\cdot)$ is an RS encoder. $A$ appends $h(k)$ to $\delta_A$ and sends the resulting value to $D$, where $h(\cdot)$ is a secure hash function.

  \item Upon receiving the commitment $\delta_A$ and $h(k)$, $D$ opens the commitment by decoding secret $k'=DEC_{RS}(F_D \ominus \delta_A)$ using its independently derived fingerprint $F_D$, where $DEC_{RS}(\cdot)$ is the RS decoding function. $D$ then verifies $h(k') \overset{?}{=} h(k)$. If the verification succeeds, $D$ accepts the secret. Otherwise, it aborts the pairing process.
\end{enumerate}

The key agreement phase is an application of a fuzzy commitment scheme \cite{juels1999fuzzy} that has been widely used in context pairing \cite{farrukh2022one, han2018you, miettinen2014context}.
It can reconcile two similar sequences that differ up to the error correction capability of the selected error correcting code.  To conceal the secret $k$, $A$ computes $\delta_A=F_A \ominus ENC_{RS}(k)$ which reveals no information about $k$ ($k$ is pseudo-random and $F_A$ is unknown to $M$.  Upon receiving the commitment $\delta_A$, $D$ can open the commitment by applying $k'=DEC_{RS}(F_D \ominus \delta_A)$ and recover $k'$.  The value $k'$ is equal to $k$ only if the Hamming distance between $F_D$ and $F_A$ is below the error-correcting capability of the RS code. Finally, the integrity of $k'$ is checked using the attached hash $h(k).$

\begin{figure}[t]
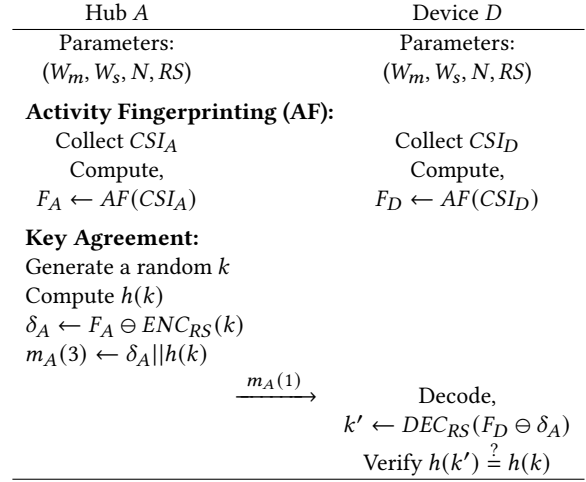

\centering
\scalebox{1}{
\begin{tabular}{ccc}
    Hub $A$& &  Device $D$\\
    \hline
    Parameters: & & Parameters:\\
    ($W_{m},W_{s},N, RS$)& &($W_{m},W_{s},N, RS$)\\[1ex]
                
    \multicolumn{3}{l}{{\bf Activity Fingerprinting (AF):}}\\
    Collect $CSI_A$ &  & Collect $CSI_D$\\
    Compute, & & Compute,\\
    $F_A\leftarrow AF(CSI_A)$ & & $F_D\leftarrow AF(CSI_D)$\\[1ex]

    \multicolumn{3}{l}{{\bf Key Agreement:}}\\
    
    \multicolumn{3}{l}{Generate a random $k$}\\
    \multicolumn{3}{l}{Compute $h(k)$}\\
    \multicolumn{3}{l}{$\delta_A\leftarrow F_A\ominus ENC_{RS}(k)$}\\
    
    $m_A(3)\leftarrow \delta_A||h(k)$ & &\\
    
    &$\xrightarrow{~m_A(1)~}$& Decode,\\ 
    
    &&$k'\leftarrow DEC_{RS}(F_D\ominus \delta_A)$\\  
    
    &&Verify $h(k^\prime)\overset{?}{=}h(k)$\\
    \hline
\end{tabular}}
\caption{The Jellybean pairing protocol.}
\label{fig:diagram}
\vspace{-0.2in}
\end{figure}

\section{Parameter Selection and Security Evaluation}

In this section, we show how to select the parameters of Jellybean to achieve correctness (secure pairing of a legitimate device) and soundness (leakage of $k$ to Mallory). We further experimentally evaluate Jellybean in terms of performance and security.

\subsection{Evaluation Metrics}
\label{sec:metrics}

To evaluate the security and performance of Jellybean, we employ the following metrics that are commonly used for context-based pairing methods.

\smallskip
 {\bf Bit Mismatch Rate (BMR):} the ratio of the Hamming distance between two fingerprints over the fingerprint length. The BMR  between $F_A$ and $F_D$ is used as a metric for protocol correctness, whereas the BMR between $F_A$ and $F_M$ is used as a metric for protocol soundness.

{\bf Secret Bit Rate (SBR):} the average number of bits extracted per any activity interrupting the $A$-$D$ channel. Here, we define the SBR per activity as opposed to per unit of time, because the activity rate could vary depending on the environment. 

{\bf Approximate Entropy (ApEn):} A measure of the randomness and unpredictability of a time series. It is preferred over entropy  because it provides a more accurate measure of randomness for finite-length series \cite{pincus1991approximate}.
 
The Jellybean protocol parameters involve tradeoffs between security and performance. For example, the RS code needs to be carefully selected to allow legitimate devices to pair but prevent Mallory from inferring the common key. We study these tradeoffs by co-evaluating correctness and soundness.

\subsection{Experiment Setup}

\begin{figure}[t]
\begin{tabular}{c}    
\includegraphics[width=0.7\columnwidth]{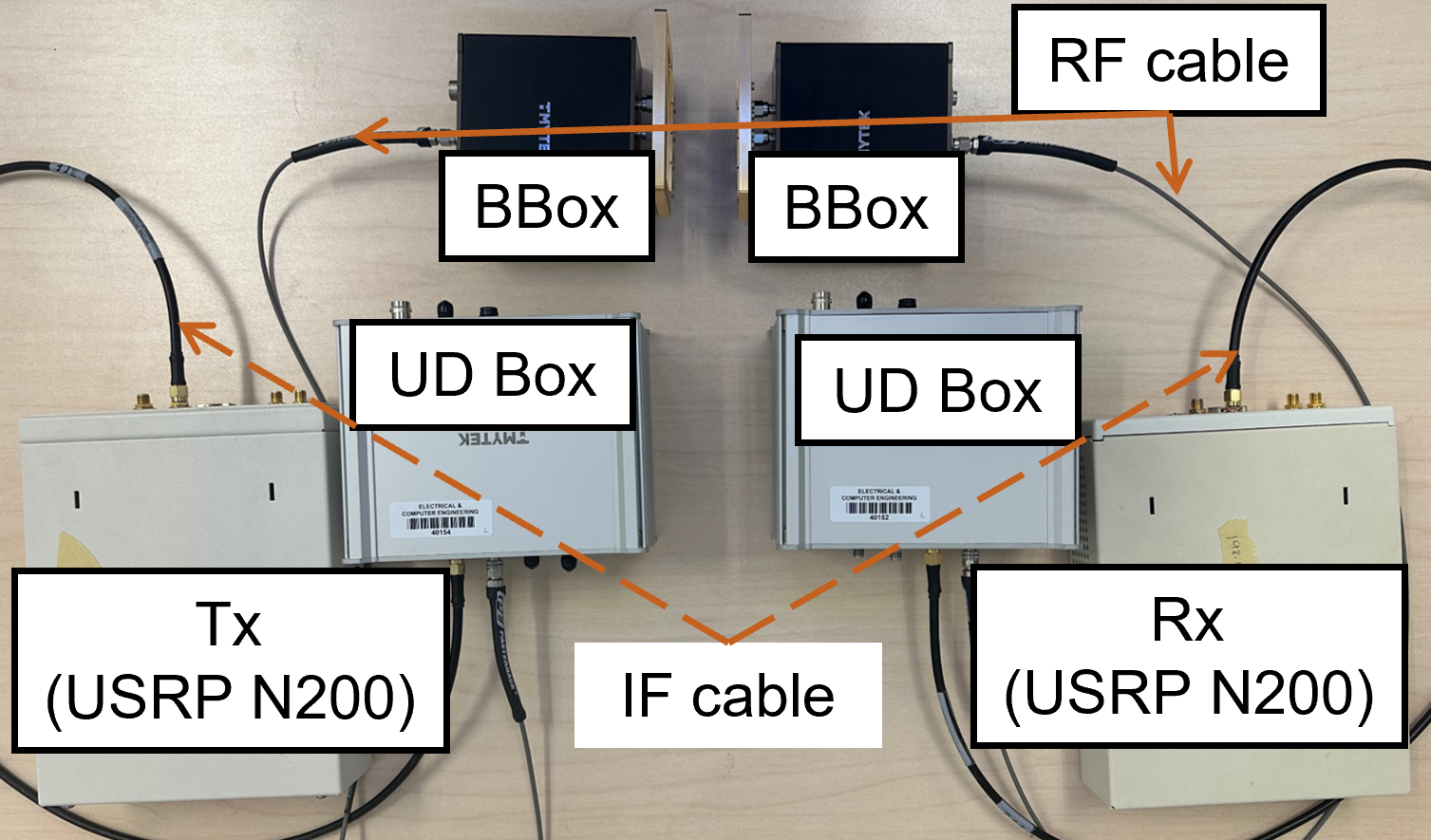}\\
(a) device connection of each link\\ \includegraphics[width=0.7\columnwidth]{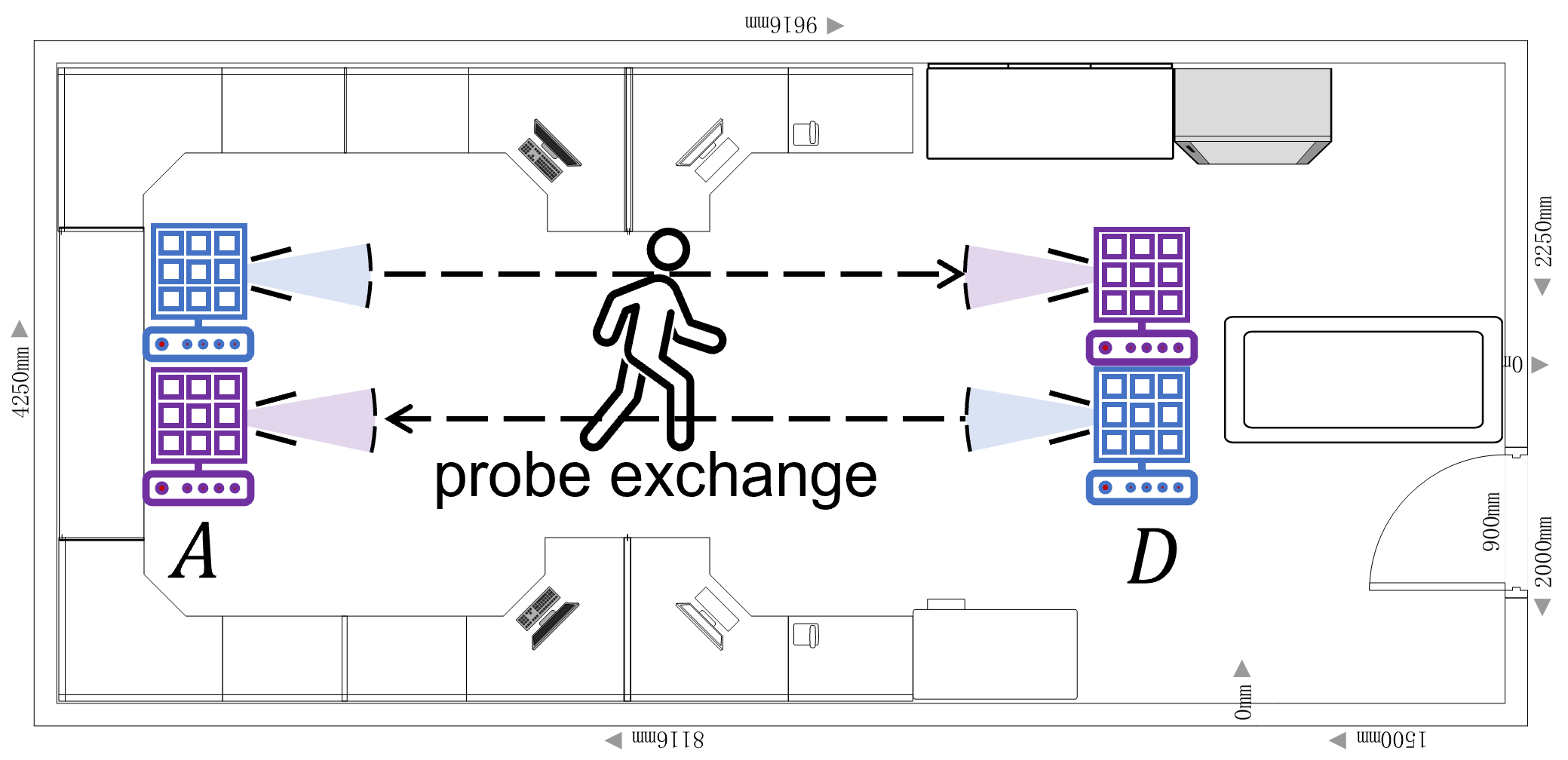}\\
(b) Room A, 438 sq.ft\\
\includegraphics[width=0.7\columnwidth]{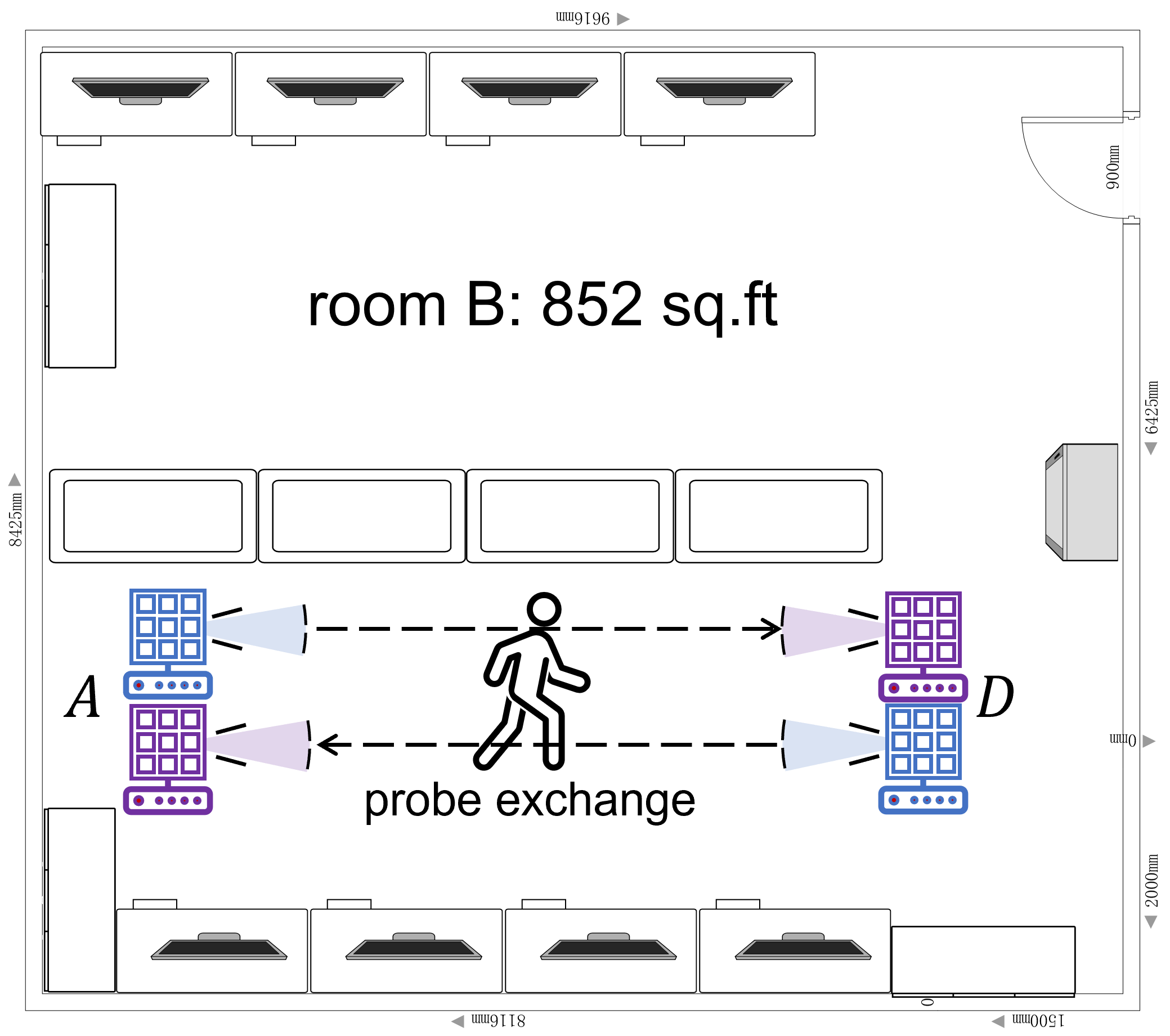} \\
(c) Room B, 852 sq.ft
\end{tabular}
\vspace{-0.1in}
\caption{Experiment setup and topologies.}
\label{fig:exp_set}
\vspace{-0.15in}
\end{figure}

Our experimental setup consisted of USRPs and  connected to the TMYTEK mmWave platform \cite{TMYTEK}. Specifically, we implemented the baseband functionality of $A$, $D$ and $M$ on USRP N200s. Each USRP was connected to a TMYTEK up-down (UD) converter that upscaled the signal from the sub-6GHz band to a center frequency of 28GHz with 160MHz bandwidth. Transmission took place via a directional BBox lite patch antenna with 16 elements. The antenna supports electronically steerable sectors with half-power beamwidth of $30^\circ$, covering the plane with 12 non-overlapping sectors. The maximum transmit power of the USRPs was 30dBm.  On the receiver side, the captured 28 GHz signal was down-converted and processed by the USRP. The platform configuration is shown in Figure~~\ref{fig:exp_set}(a).

The exchanged probes consisted of 1,460 bytes and were transmitted using a 5MHz bandwidth over 52 subcarriers. The sampling rate for the USRPs was set to 5 MHz enabling the collection of 3,100 CSI samples for each subcarrier per second. The probing process for CSI measurements requires rapid and continuous switching between signal transmission and reception. Due to the half-duplex capabilities of  USRPs N200 and their slow switching capabilities, we performed the probing  between $A$ and $D$ using two horizontally-aligned links, as shown in  Figures~\ref{fig:exp_set}(b) and \ref{fig:exp_set}(c), where each mmWave communication link consists of the setting in Figure~\ref{fig:exp_set}(a). Figures~\ref{fig:exp_set}(b)  and \ref{fig:exp_set}(c) also show the layout of the two office rooms where we conducted our experiments. 
 
We studied both artificial and natural motions within the environment. To create artificial motions in the $A$-$D$ channel, we waived a hand vertically across one of the antennas to disrupt the channel. The hand was waived at natural speeds and also horizontally to disrupt any paths beyond the LoS. The artificial motion reflects a scenario where a user is instructed to waive his hand across the device that he intends to pair and provides a usable and quick way for device pairing. 

For natural motions, the office occupants walking through the space disrupting different paths at different times. The natural motion reflects a scenario where the natural movement in the space is exploited to pair device as well as continuously renew trust.

\subsection{Parameter Selection}
\label{sec:parameter}

\begin{figure}[t]
\centering
\setlength{\tabcolsep}{-5pt}
\begin{tabular}{cc}
\includegraphics[width=0.55\columnwidth]{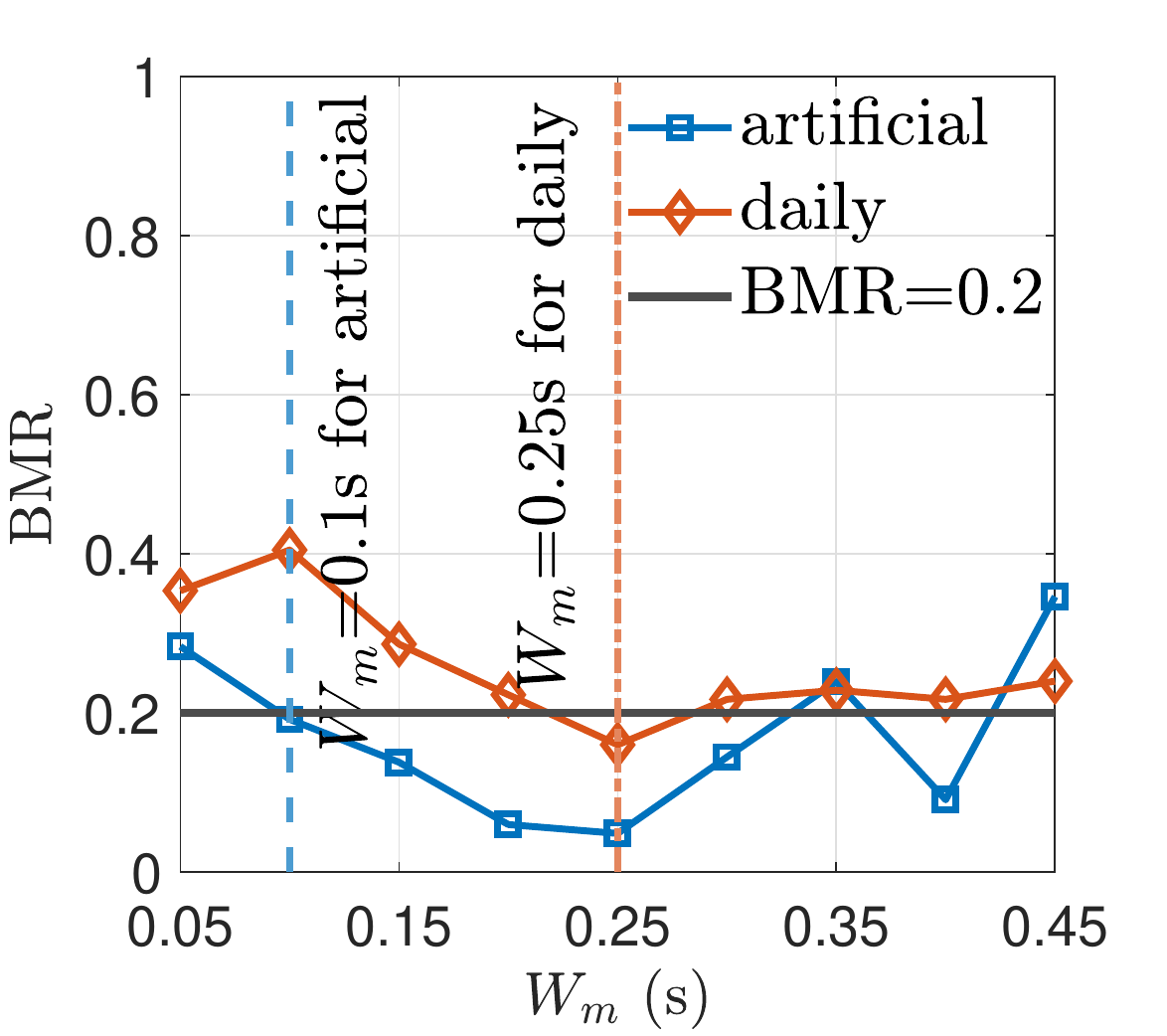}&
\includegraphics[width=0.55\columnwidth]{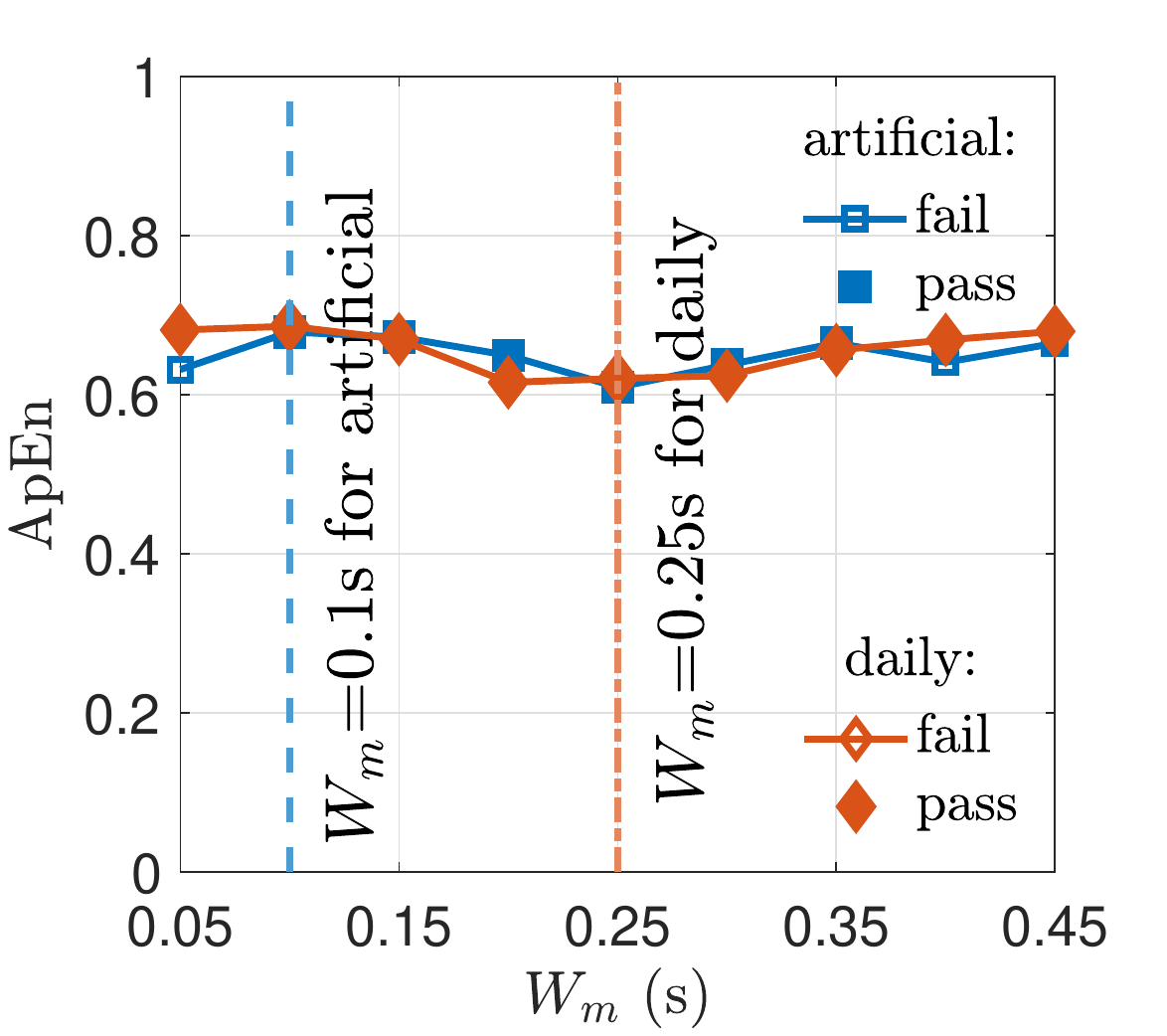} \\
(a) & (b)
\end{tabular}
\vspace{-0.1in}
\caption{ (a) BMR as a function of $W_m$ and (b) ApEn as a function of  $W_m$.}
\label{fig:parameter1}
\vspace{-0.2in}
\end{figure}

\begin{figure*}[t]
\centering
\setlength{\tabcolsep}{-7pt}
\begin{tabular}{ccc}  
  \includegraphics[width=0.78\columnwidth]{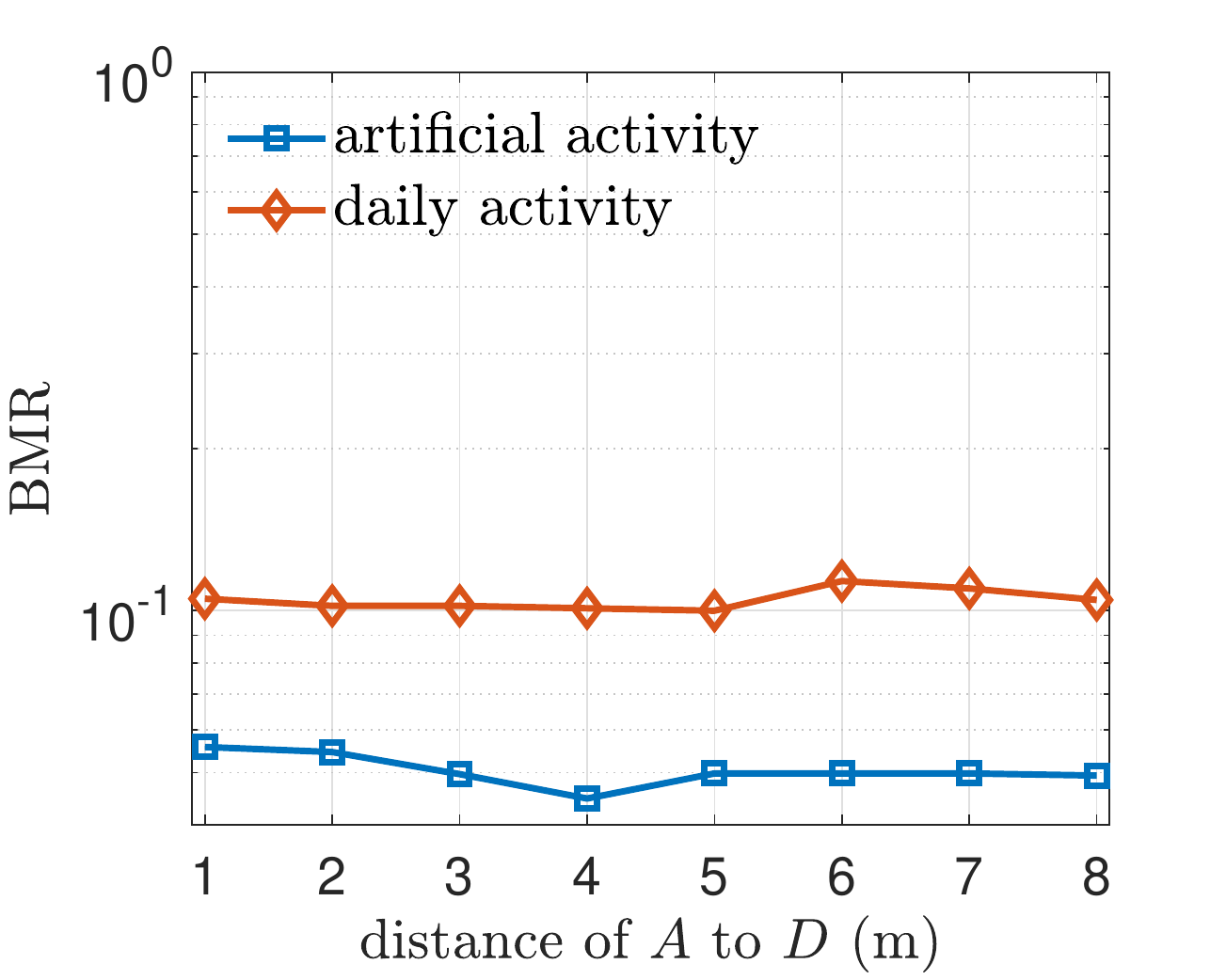}&
  \includegraphics[width=0.78\columnwidth]{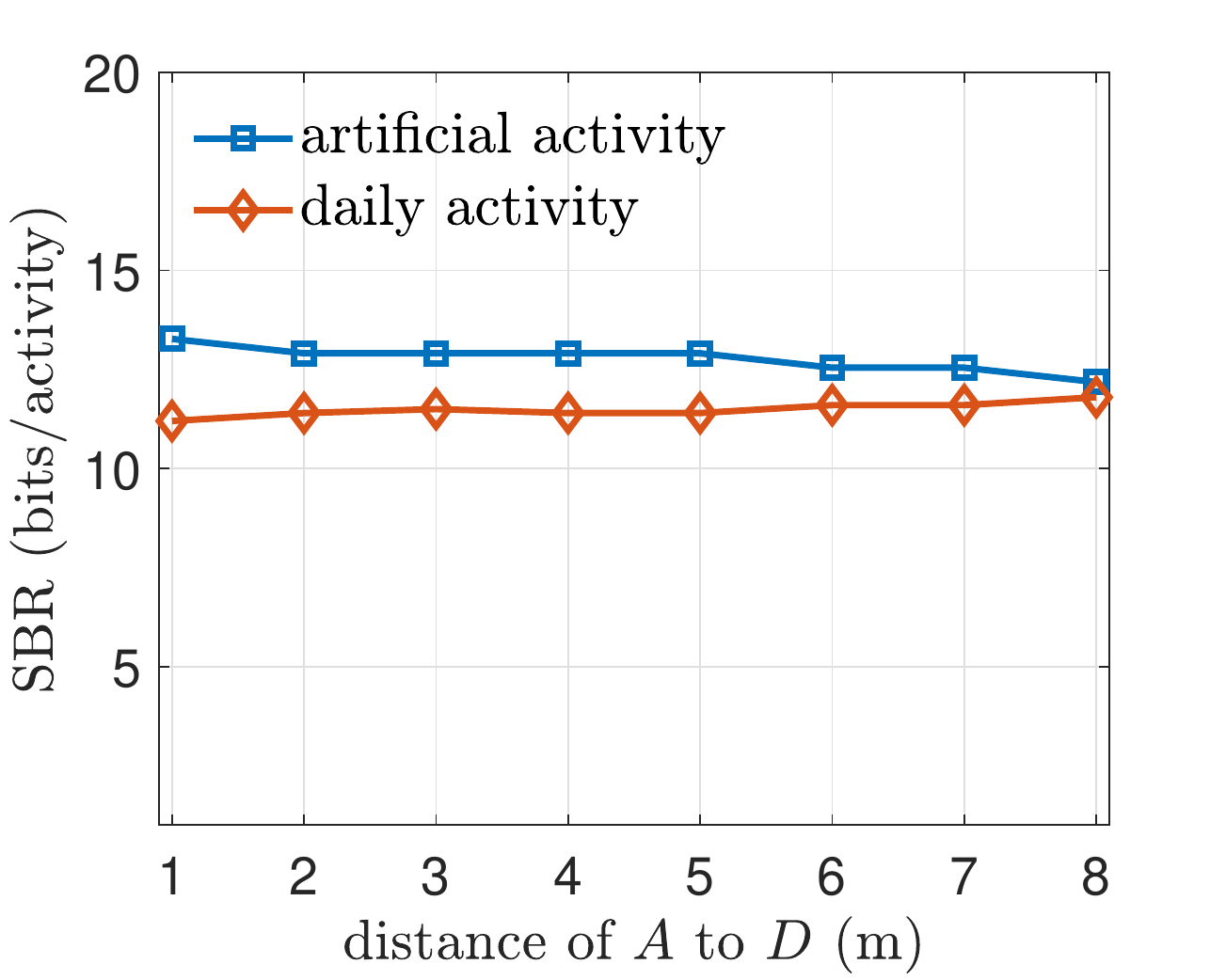}&
  \includegraphics[width=0.78\columnwidth]{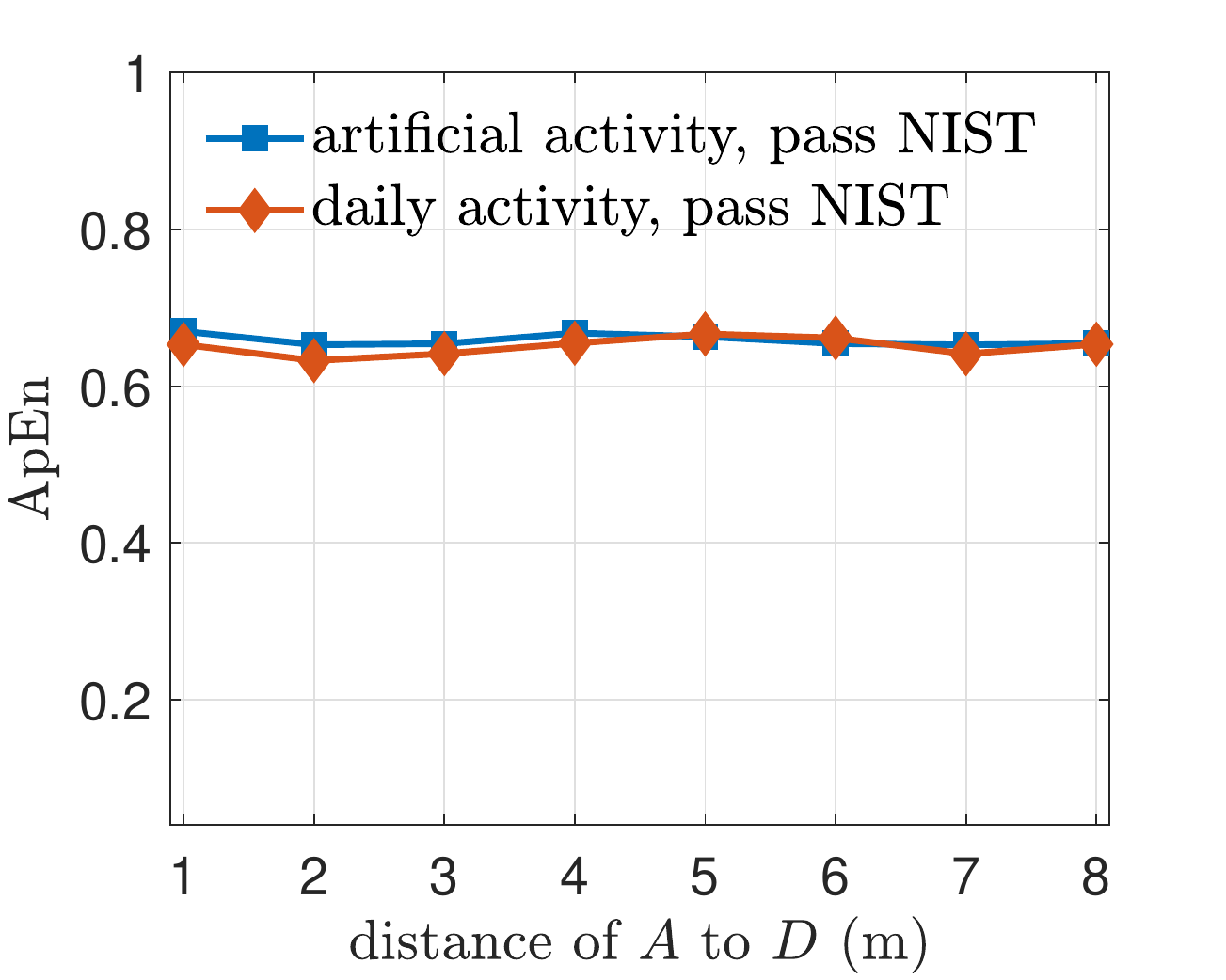}
\end{tabular}
\vspace{-0.15in}
\caption{Performance evaluation of Jellybean with two types of activities in room B.}
\label{fig:dis_path}
\vspace{-0.1in}
\end{figure*}

The Jellybean protocol is controlled by three parameters: the moving variance window $W_{m}$, the averaging window $W_{s}$,  the number of retained LSBs $N$ in each RLE block, and the RS code for key agreement. In this section, we show how
to select these parameters in practice and then evaluate the Jellybean protocol under different attack scenarios.  The CSI samples used for parameter selection are collected in 90 second intervals, during which the LoS path of $A$-$D$ channel in room A is interrupted by the occupant with artificial and daily activities. The approximate entropy of the fingerprint is calculated using the NIST randomness test suite \cite{bassham2010sp}. Moreover, we obtained the $p$-value from the NIST tests which indicates if a sequence is sufficiently random to be used for security purposes. The $p$-value must exceed 0.01. For all the ApEn plots, we utilize empty markers to represent results with a $p$-value lower than 0.01, which indicates a failure in the NIST randomness test. Conversely, filled markers indicate a successful test.

{\bf Selecting $W_{m}$.} The moving variance window $W_m$ alleviates the high sensitivity of the CSI amplitude measurements. It needs to be long enough to smooth the rapid fluctuations, but should not eliminate  the CSI randomness. Figure~\ref{fig:parameter1} presents the BMR and ApEn as a function of $W_{m}$ for both artificial and natural daily activities. To simplify the analysis, $W_{m}$ is converted from the number of CSI samples to seconds based on the sampling rate. We choose the smallest $W_{m}$ value that results in a BMR of less than 0.2 and an ApEn with $p$-values exceeding 0.01. As a result, we set $W_{m}$ to 0.1s for artificial activities and 0.25s for daily activities. This selection balances the functionality and security of the Jellybean protocol while ensuring a sufficiently random fingerprint is generated.

\begin{figure}[t]
\centering
\setlength{\tabcolsep}{-5pt}
\begin{tabular}{cc}
\includegraphics[width=0.55\columnwidth]{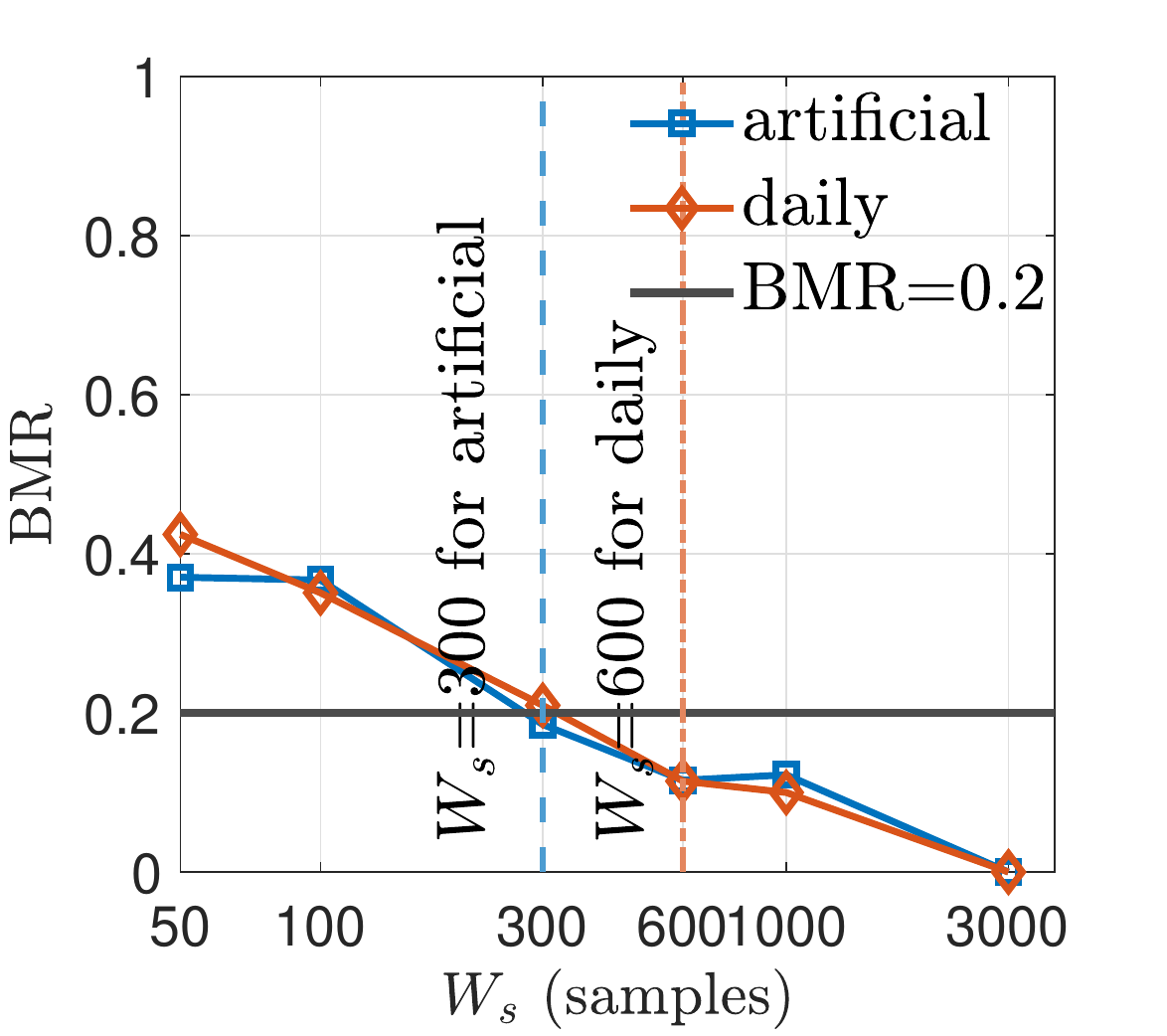}&
\includegraphics[width=0.55\columnwidth]{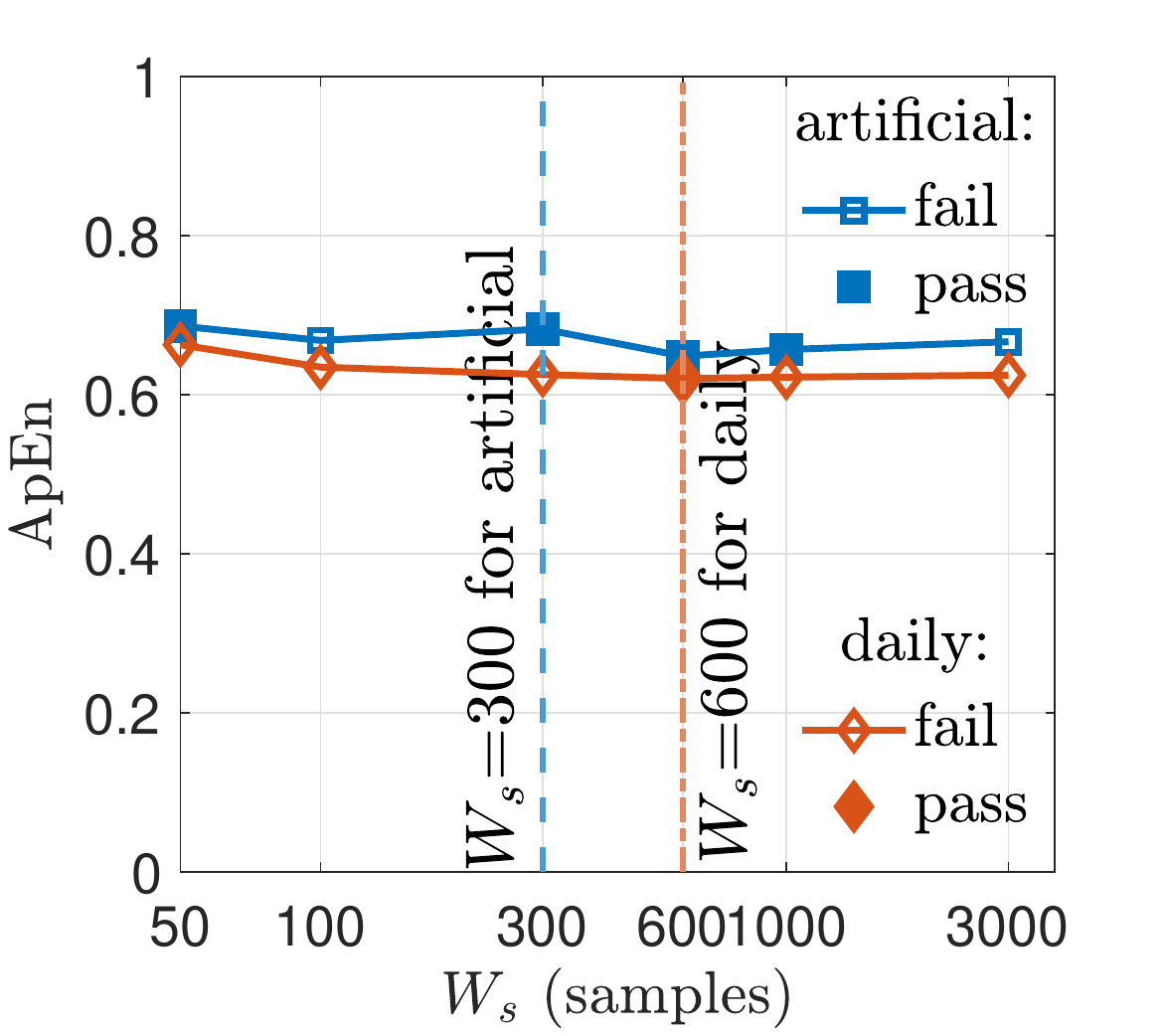} \\
(a) & (b)
\end{tabular}
\vspace{-0.1in}
\caption{ (a) BMR as a function of $W_s$, and (b) ApEn as a function of $W_s$.}
\label{fig:parameter2}
\vspace{-0.2in}
\end{figure}

{\bf Selecting $W_{s}$.} The averaging window $W_{s}$ provides an additional level of smoothing to account for imperfect channel reciprocity and the sensitivity of the mmWave channel to small changes in geometry. It is also of great significance to avoid prolonging $W_{s}$ that can result in a predictable fingerprint. In Figure~\ref{fig:parameter2}, we show the BMR and ApEn as a function of $W_s$. We observe that the BMR monotonically decreases as $W_s$ increases. However, compared to $W_m$, the Jellybean protocol is more sensitive to changes in $W_{s}$, and the generated fingerprint can easily become predictable, failing the NIST randomness test for many $W_{s}$ values. We choose the smallest $W_{s}$ value that yields a BMR below 0.2 and an ApEn with a $p$-value greater than 0.01, indicating a usable fingerprint. Consequently, we let $W_s=300$ and $W_s=600$ for artificial and daily activities, respectively.

\begin{figure}[t]
\centering
\setlength{\tabcolsep}{-5pt}
\begin{tabular}{cc}
\includegraphics[width=0.55\columnwidth]{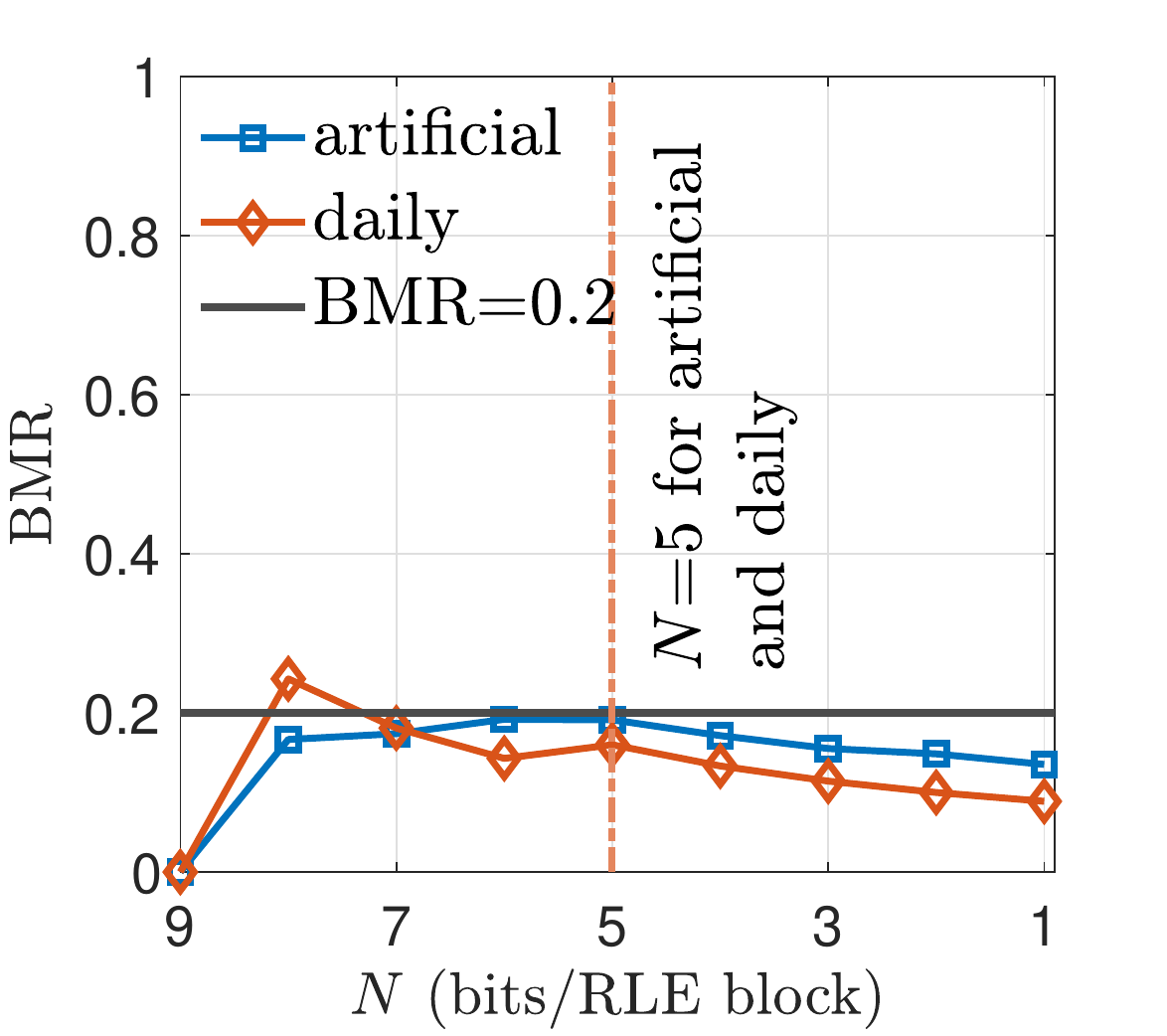}&
\includegraphics[width=0.55\columnwidth]{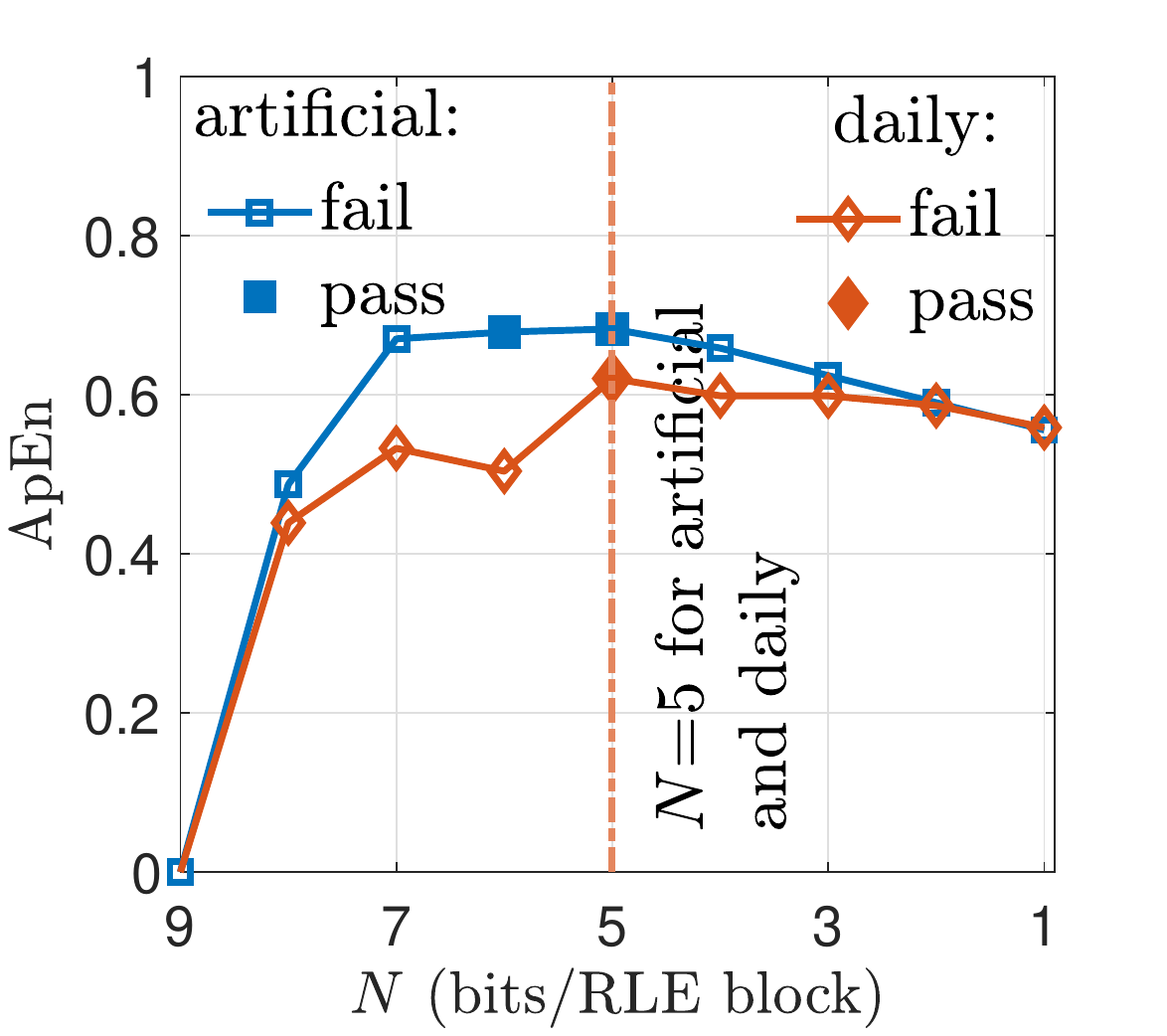} \\
(a) & (b)
\end{tabular}
\vspace{-0.1in}
\caption{(a) BMR as a function of $N$, and (b) ApEn as a function of $N$.}
\label{fig:parameter3}
\vspace{-0.2in}
\end{figure}
{\bf Selecting $N$.} The number of LSBs maintained from each CSI block control the contribution of each event (activity or inactivity) to the final bit sequence used as fingerprint. The selection of $N$ also impacts the entropy of the final bit sequence as a longer $N$ could lead to more predictable sequences. Figure~\ref{fig:parameter3} shows how BMR and ApEn vary with $N$. As expected, a smaller $N$ leads to a higher BMR but also a higher ApEn. To best balance the functionality and security of the Jellybean protocol, we fix $N=5$ for both artificial and daily activities.

{\bf Selecting the RS code.} The final parameter to be selected is the RS code used in the key agreement phase. The code's error correcting capability must be high enough to allow for the observed BMR, but low enough so that $k$ is not leaked to Mallory. Based on the previously selected parameters and the obtained BMR values, we considered an RS code that tolerated a 20\% mismatch between the two fingerprints. This is inline with codes used in previous context-based methods.

\subsection{Performance Evaluation of Jellybean}
In this section, we evaluate the performance of Jellybean protocol using the parameters set above. To test the protocol robustness to environmental changes, we conducted additional experiments in room B. Room B has a larger space and a more complicated RF environment compared to room A. The layout of Room B is shown in Figure \ref{fig:exp_set}(c).
We conducted additional experiments in room B, shown in Figure~\ref{fig:exp_set}(c), to illustrate varying $A$-$D$'s distance impact on evaluation metrics BMR, SBR, and ApEn. We varied the distance between $A$ and $D$ from 1m to  8m. For each distance, we collected CSI samples for three minutes while performing artificial and natural daily activities. Figure~\ref{fig:dis_path} shows the BMR, SBR, and ApEn at different distances. We observe a sound trend of the three metrics, which indicates that the $AF$ algorithm can develop activity-based fingerprints with a low BMR, an acceptable SBR, and NIST-approved ApEn.

\subsection{Security Analysis for Jellybean}
\label{sec:security}
In this section,we analyze the security of our Jellybean protocol against different types of adversaries.

\begin{figure}[t]
\centering
\setlength{\tabcolsep}{2pt}
\begin{tabular}{cc} 
\includegraphics[width=0.5\columnwidth]{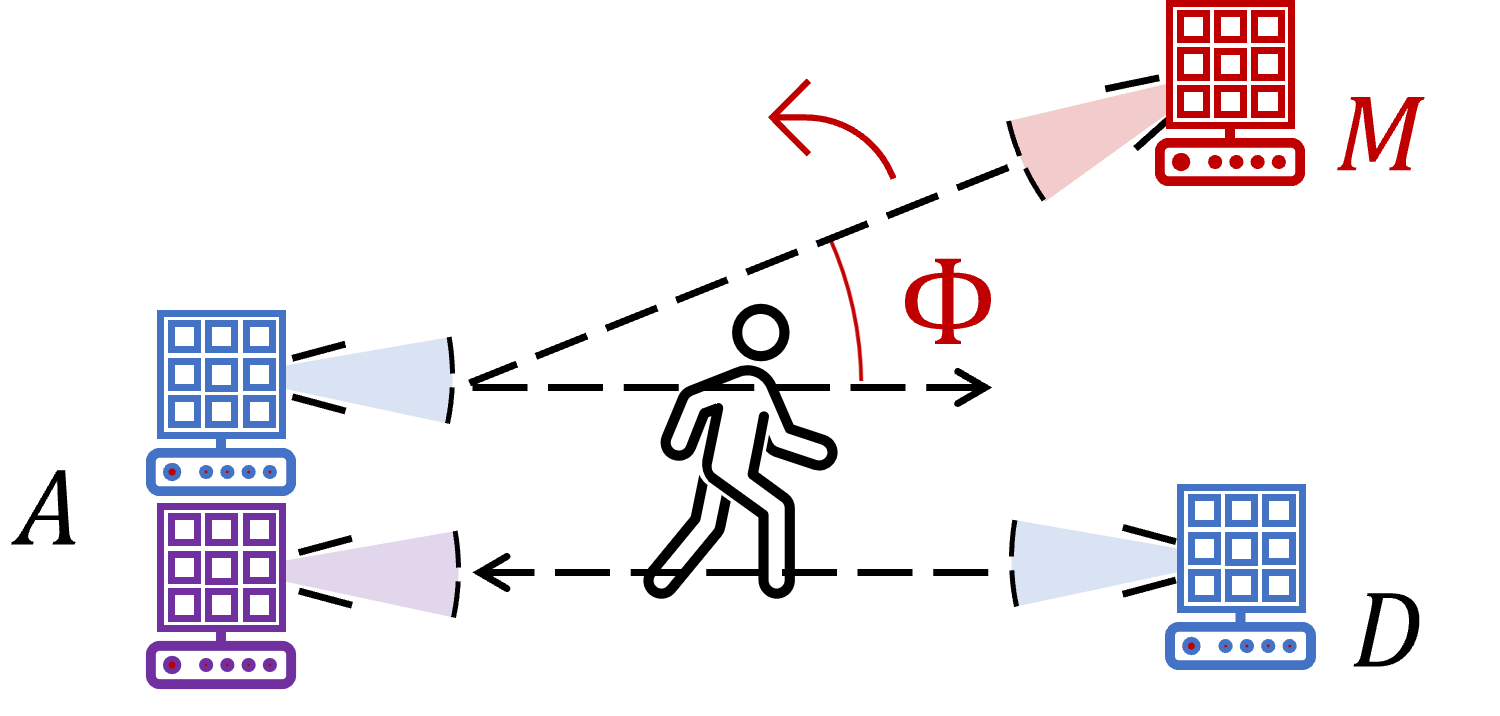}& 
\includegraphics[width=0.5\columnwidth]{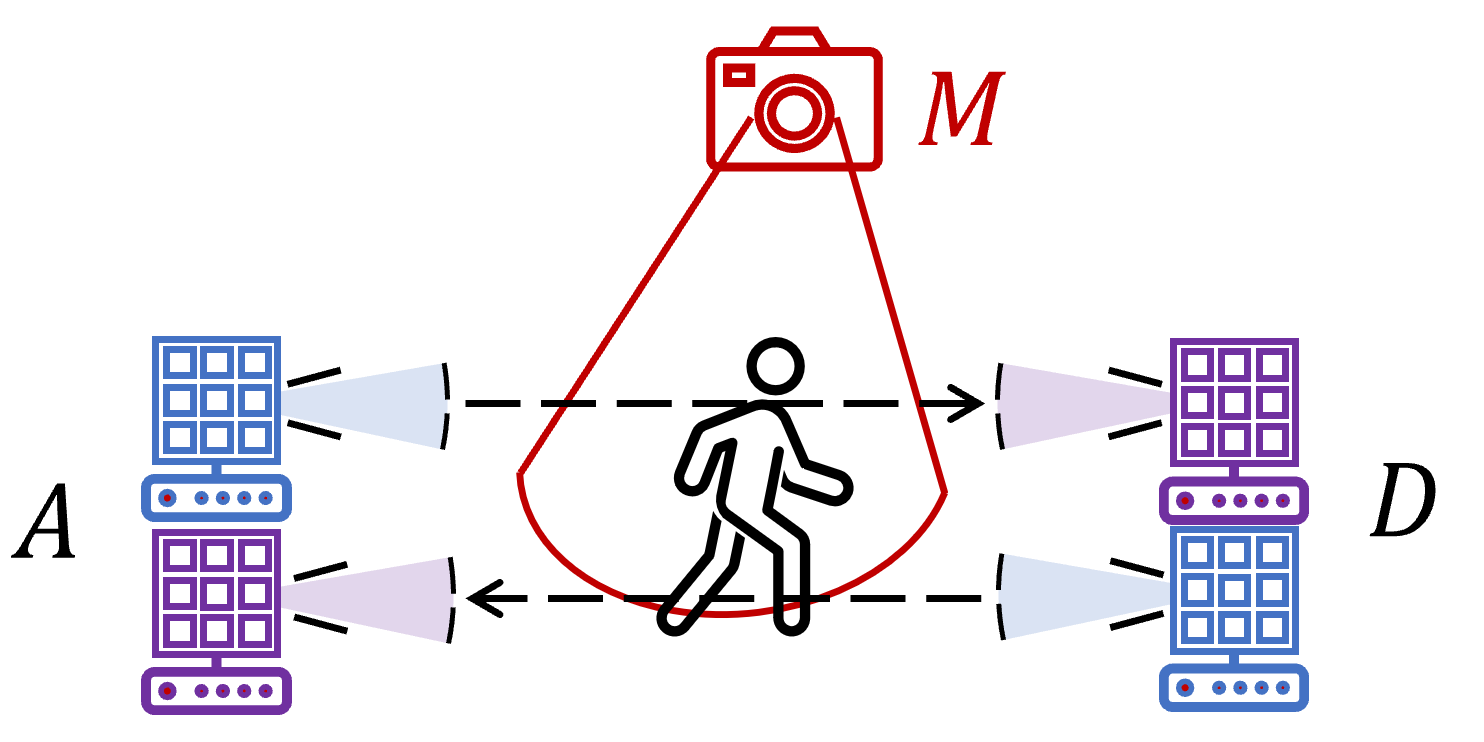}\\
(a) an eavesdropping $M$ &
(b) a keylogging $M$\\[1ex]
\multicolumn{2}{c}{\includegraphics[width=0.7\columnwidth]{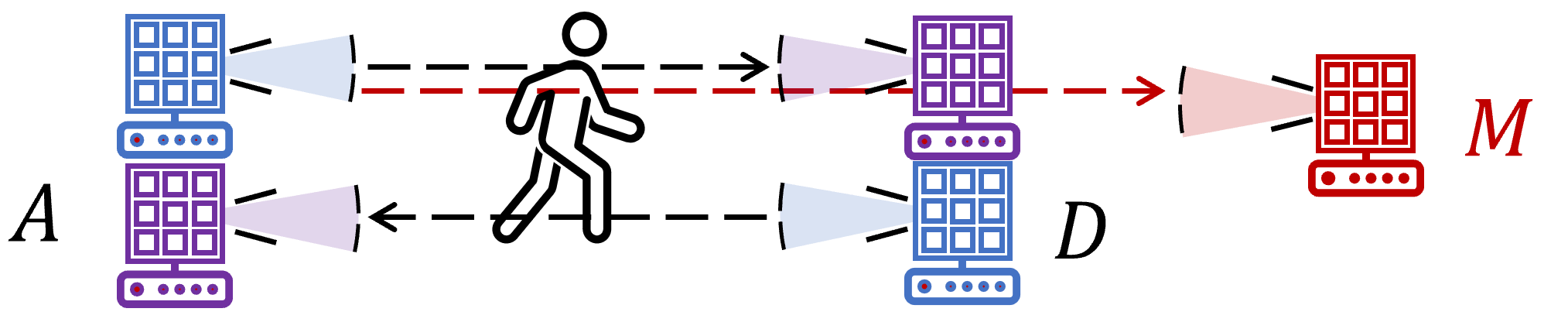}}\\
\multicolumn{2}{c}{(c) a nearby eavesdropper $M$}\\
\end{tabular}
\vspace{-0.1in}
\caption{Experiment settings for evaluating Jellybean with different types of $M$.}
\label{fig:M_set}
\vspace{-0.15in}
\end{figure}

\begin{figure}[t]
\centering
\setlength{\tabcolsep}{-5pt}
\begin{tabular}{cc} 
  results in room A & results in room B\\[1ex]
  \includegraphics[width=0.55\columnwidth]{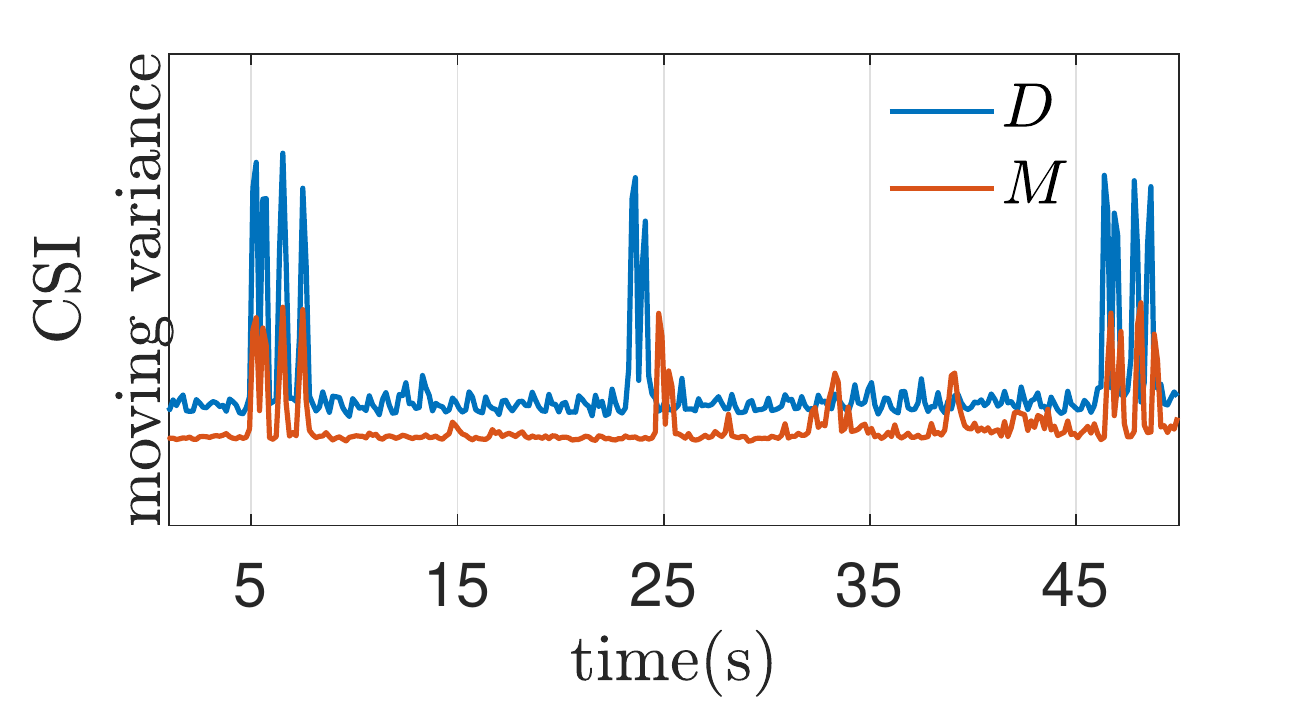}&
  \includegraphics[width=0.55\columnwidth]{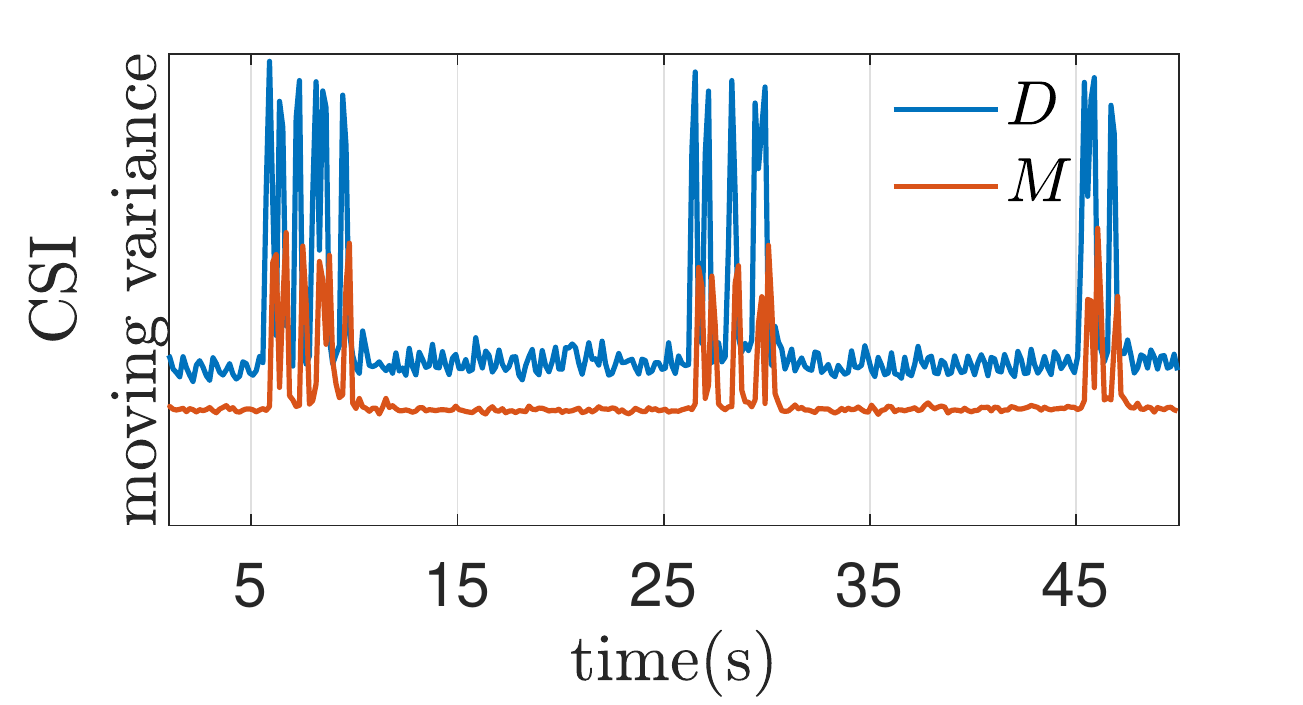}\\
\multicolumn{2}{c}{(a) $M$ with an intersection angle $\Phi=0^{\circ}$}\\
  \includegraphics[width=0.55\columnwidth]{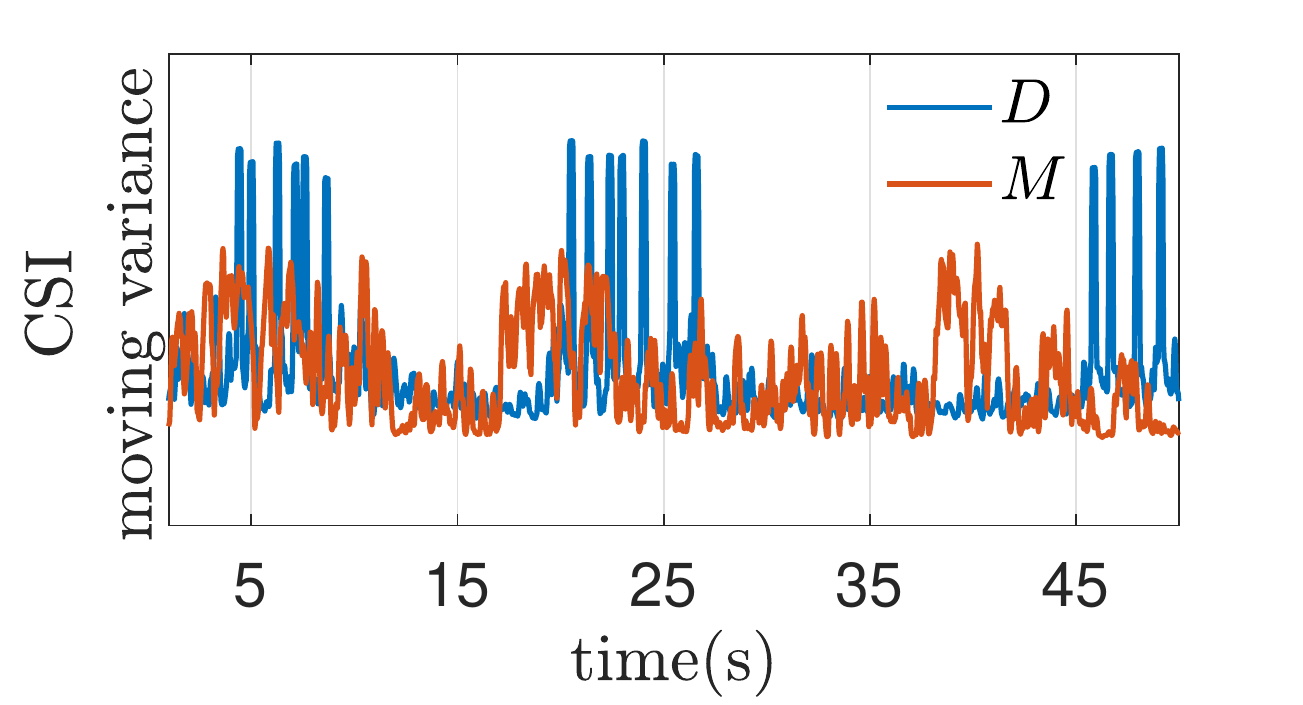}&
  \includegraphics[width=0.55\columnwidth]{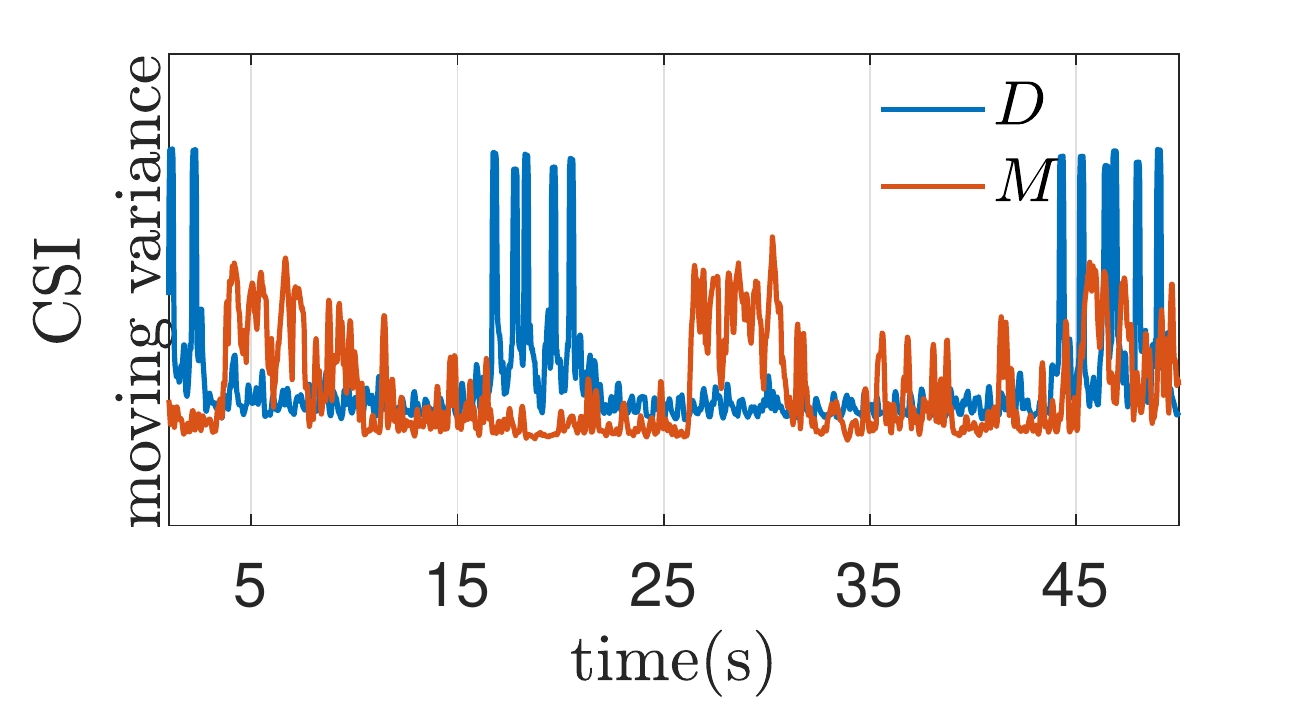}\\
\multicolumn{2}{c}{(b) $M$ with an intersection angle $\Phi=20^{\circ}$}\\
  \includegraphics[width=0.55\columnwidth]{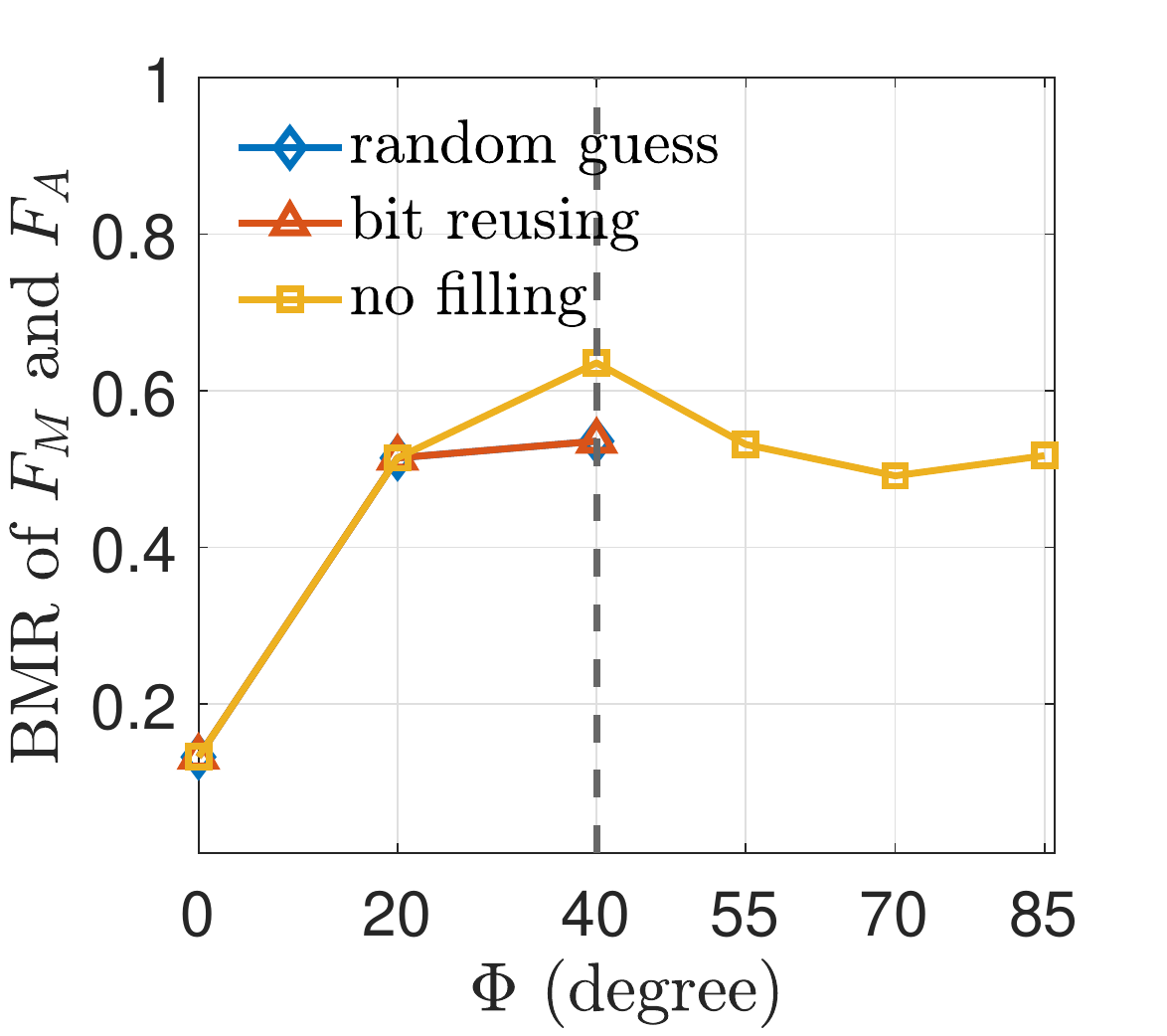}&
  \includegraphics[width=0.55\columnwidth]{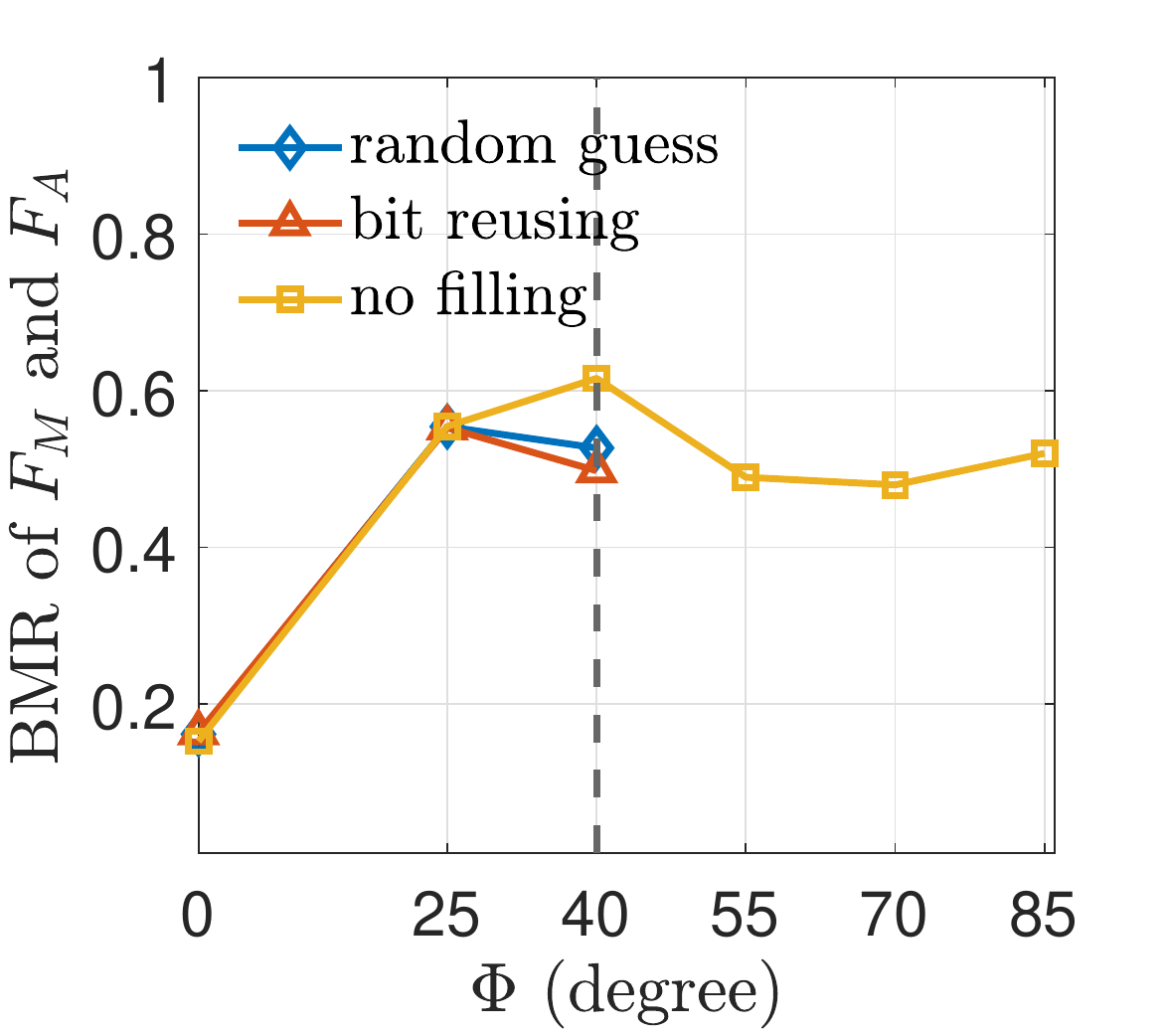}\\[1 ex]
\multicolumn{2}{c}{(c) BMR of $F_A$ and $F_M$ as a function of $\Phi$, $F_M$ after the dash line}\\
\multicolumn{2}{l}{\hspace{.6cm} is generated with random guess due to the trivial CSI samples}\\
\multicolumn{2}{l}{\hspace{.6cm} received at $M$.}
\end{tabular}
\vspace{-0.1in}
\caption{Evaluation of Jellybean for an eavesdropping $M$ in room A (the left column) and room B (the right column).}
\label{fig:bms_M}
\vspace{-0.15in}
\end{figure}

\begin{figure}[t]
\begin{tabular}{c}    
\includegraphics[width=1.1\columnwidth]{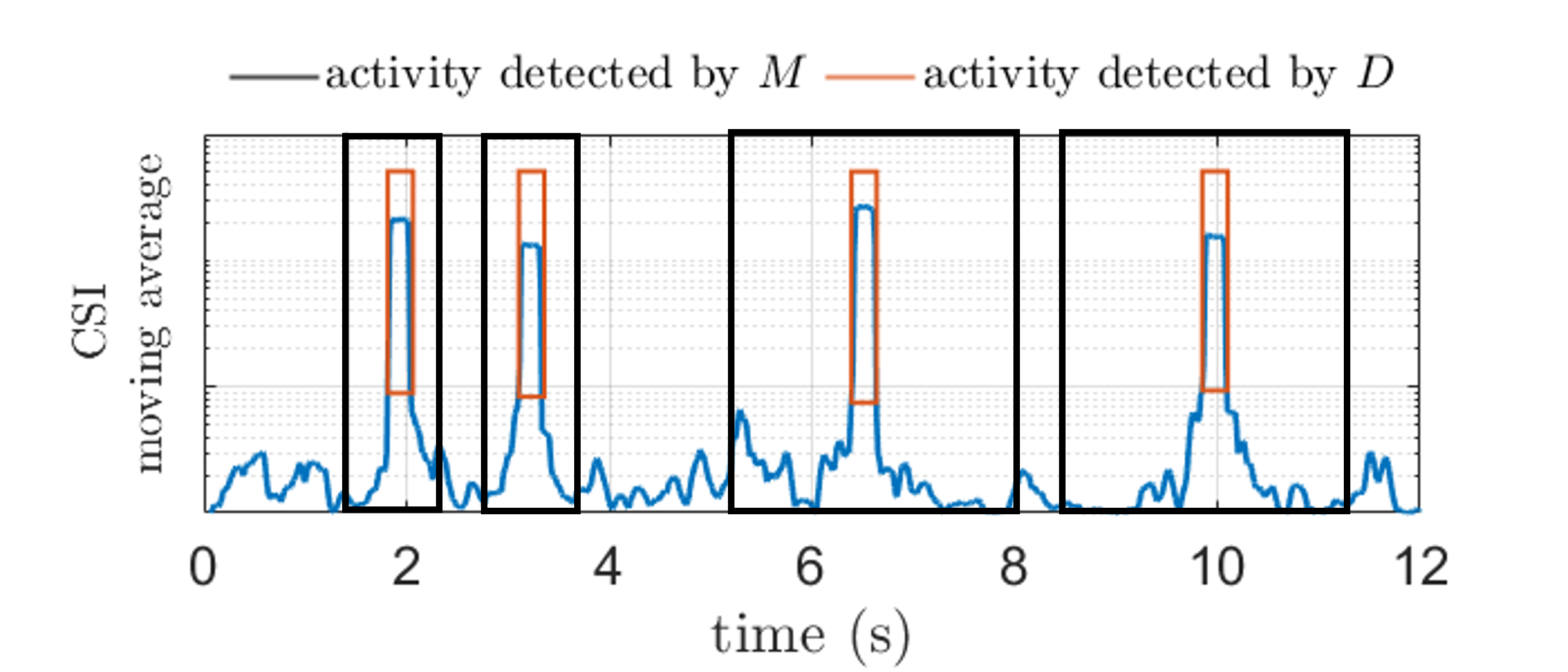} \\
(a) artificial activity\\
\includegraphics[width=1.1\columnwidth]{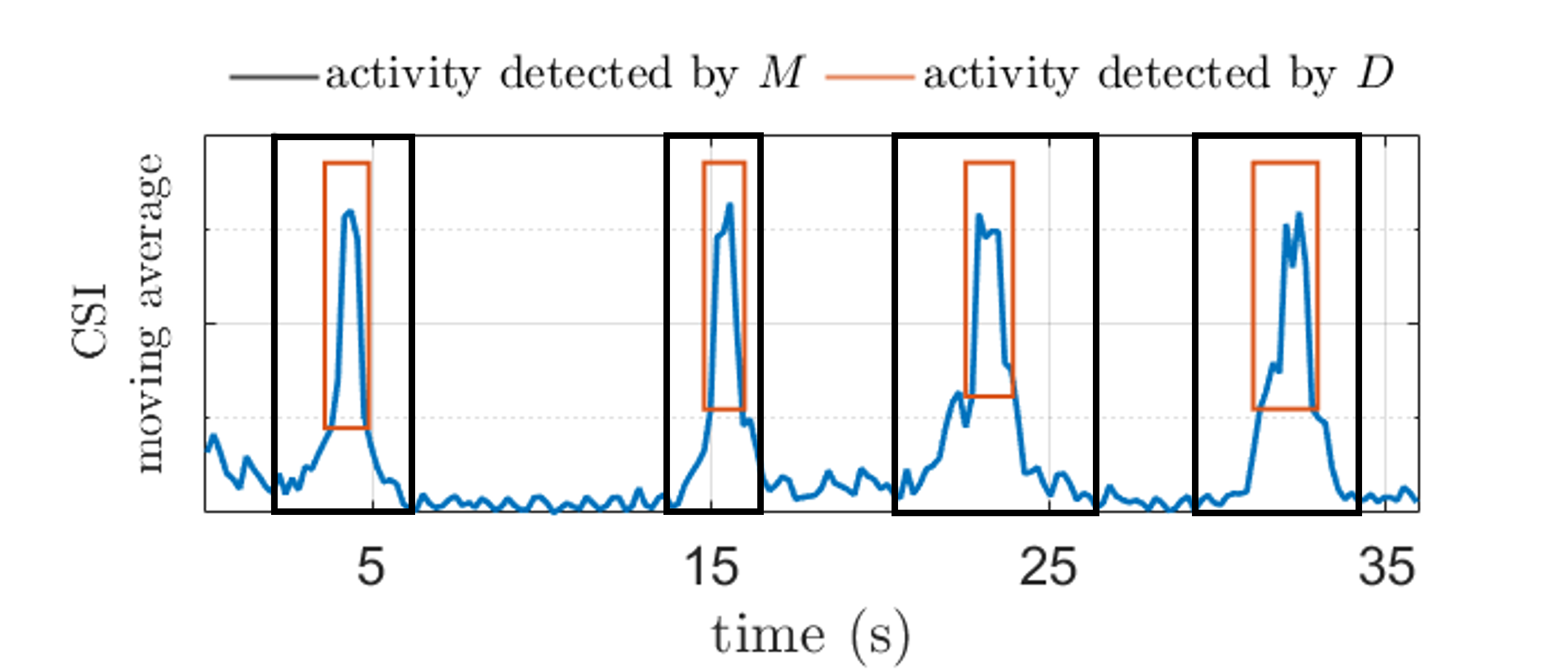} \\
(b) daily activity
\end{tabular}
\vspace{-0.1in}
\caption{Security evaluation of the Jellybean for a keylogging $M$ with two types of activities in experimental room A.}
\label{fig:motion_detec}
\vspace{-0.15in}
\end{figure}

\begin{figure*}[t]
\centering
\begin{tabular}{c}
\includegraphics[width=2.1\columnwidth]{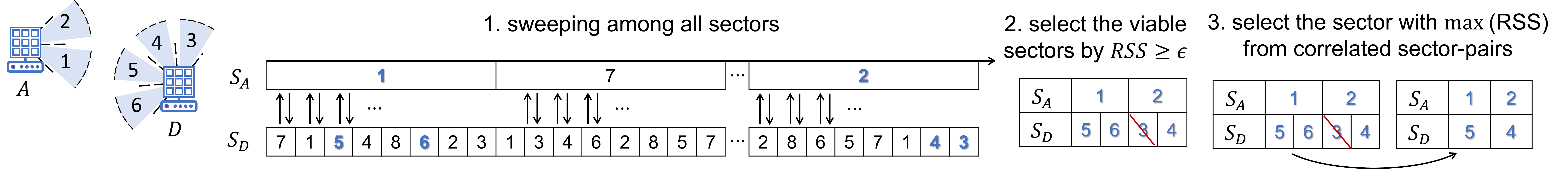}\\
(a) path discovery\\
\includegraphics[width=2.1\columnwidth]{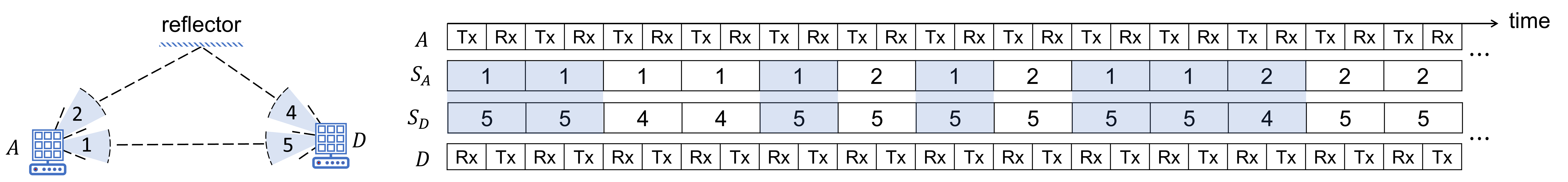}\\
(b) uncoordinated path hopping
\end{tabular}
\vspace{-0.1in}
\caption{An example of the path discovery phase and path hopping phase in UPH. The numbers indicate the chosen sectors at each device. If $A$ and $D$ simultaneously switch to a usable sector pair (2-4 and 1-5 in the example), the probing round is successful and activities are measured.}
\label{fig:uncoordiated_beamforming}
\vspace{-0.15in}
\end{figure*}

{\bf 1) Eavesdropping Adversary:} We first consider an eavesdropper $M$ is positioned at the same distance from $A$ as $D$. This positioning is intended to ensure a similar channel loss for $A$-$M$ channel as for $A$-$D$ channel, thereby enabling $M$ to actively adjust its angle from the $A$-$D$ channel to identify the best eavesdropping opportunities. Due to the hardware limitation, we only used four USRPs N200 to implement $A$, $D$, and $M$, as shown in Figure \ref{fig:M_set}(a). We maintained the distance for $A$-$D$ and $A$-$M$ links as 4m and collected 2-minute CSI samples with different angles $\Phi$ between $A$-$D$ and $A$-$M$ channels. To ensure fairness, activities were conducted near hub $A$, allowing activities observations for both $D$ and $M$.

Due to the directionality of the directional beam produced by the phased array beamformer, the probing signal emitted by $A$ gets lost at $M$ when $\Phi$ exceeds a certain value. To compensate the sample loss and compare the fingerprint of $M$ with that of $A$, we fill in the corresponding missing bits in $F_M$ to ensure that $F_A$ and $F_M$ have the same length. We employed three strategies for this purpose: 1) random guess that involves randomly filling in a bit to replace the missing one in $F_M$, 2) bit reusing that reuses a segment of previously extracted bits from $F_M$, and 3) no filling that leaves the missing bit blank. Figure \ref{fig:bms_M}(c) shows the BMR between $F_A$ and $F_M$ as a function of $\Phi$ for these three strategies. We observed a high BMR between $F_A$ and $F_M$, except when $\Phi=0^{\circ}$.  Figures \ref{fig:bms_M}(a) and (b) further present the moving variance of the 45-second CSI samples collected by $D$ and $M$, respectively, for $\Phi=0^\circ$ and $\Phi=20^\circ$. When $\Phi=20^\circ$, there is a significant difference in observations between $D$ and $M$, resulting in a failed key agreement process at $M$. However, when $M$ perfectly aligns with the $A$-$D$ link at $\Phi=0^\circ$, she can successfully retrieve the agreed key between $A$ and $D$ based on her own observations, which exhibit a relatively low BMR of 0.2. Such a strong eavesdropper is referred to as a {\em nearby} eavesdropper and we introduce an upgraded protocol called Jellybean+ in Sec.~\ref{sec:uncoor_protocol} to thwart it. In summary, the basic Jellybean protocol is secure against  a co-located eavesdropper outside the close proximity to $A$ or $D$.

{\bf 2) Keylogging Adversary:} A keylogging adversary uses vision instead of RF devices for activity fingerprinting. $M$ monitors all the activities taking place in the shared area between $A$ and $D$ with a camera, as shown in Figure~\ref{fig:M_set}(c). She determines the start and ending time for each activity by examining all the video recordings frame by frame and fingerprints the activity by filling bits 1s in the duration of each activity and 0s for the rest.

We let the keylogging adversary record the activities with a cell phone and we evaluate the performance of it by comparing the start and end time of the activities extracted by $M$ and $D$. Figure~\ref{fig:motion_detec} shows the activity detection results for $D$ and $M$ with one artificial and daily activity, respectively. We can see that due to the denoising and smoothing process in the Jellybean protocol,  $D$ detects the duration of one artificial activity in a range of 0.08-0.1 seconds, whereas $M$ records a longer duration spanning 1-3 seconds for the same activity. Similarly, for daily activities, $M$ detects the activity with a significantly longer duration compared to $D$. This substantial discrepancy in activity fingerprinting between $M$ and $D$ makes it impossible for $M$ to agree on the same key as established between $A$ and $D$.

{\bf 3) Beam-stealing Adversary.} $M$ deceives $A$ and $D$ by manipulating the beam alignment process. $M$ forges sector sweep frames with high received signal strength (RSS), leading $A$ and $D$ to align their beams with $M$ instead of each other \cite{steinmetzer2018beam}. This beam-stealing attack allows $M$ to launch a MitM attack against the pairing protocol. 

The basic Jellybean protocol is vulnerable to this attack due to the lack of authentication in the beam alignment process. With the amended Jellybean+ protocol introduced in Sec. \ref{sec:uncoor_protocol}, we remove the assumption of beam alignment and instead employ an uncoordinated path-hopping (UPH) mechanism. This approach takes advantage of multiple paths and utilizes the electronic antenna steering capabilities of mmWave technology to randomize the activity sensing area. By adopting Jellybean+, we not only defend against the strong nearby eavesdropper, but we also mitigate the threat posed by an active beam-stealing adversary.

\section{Uncoordinated Path Hopping}
\label{sec:uncoor_protocol}
Based on the security analysis, the activity fingerprinting process is vulnerable to two attacks. First, a passive adversary that is aligned with the path observed by $A$ and $D$ will extract a similar fingerprint to $F_A$ and $F_D$ (with a BMR below 0.2 as seen in Figure~\ref{fig:bms_M}(c)), thus also being able to obtain $k$ from the commitment value $\delta_A.$ Second, an active adversary launching a beam-stealing attack can become a MitM between $A$ and $D$. In this scenario, $M$ can independently establish keys with $A$ and $D$ using the activities in the $A$-$M$ and the $M-D$ paths, respectively. 

To prevent these attacks, we propose the {\em Jellybean+} protocol that exploits the potential availability of multiple communication paths between $A$ and $D$ (e.g., LoS path and one reflection path). The main idea is to ``hop'' (steer) the antennas between paths and merge activities from each path (direction) into the same fingerprint. Without knowing which path to monitor at any given time, Mallory cannot extract the same fingerprint, even if she can observe all paths. Likewise, in a beam-stealing attack scenario, the activities toward the direction of Mallory contribute to the fingerprint only partially, since activities on other paths are also present. Evidently, path hopping is similar to frequency hopping but applied in the space domain as opposed to the frequency domain.

There are two main challenges for applying path hopping in our context, due to the lack of initial trust. First, how do $A$ and $D$ discover viable communication paths (paths where probe decoding is possible) and second, how do $A$ and $D$ securely share a common path hopping sequence. We address both challenges by developing an {\em uncoordinated path-hopping mechanism} (UPH), where $A$ and $D$ hop independently between the available paths. The UPH mechanism has two phases; the path discovery phase and the path hopping phase.

{\bf 1. Path discovery.} To discover the available paths between $A$ and $D,$ we propose a new beam alignment process that keeps track all transmitter-receiver antenna sector combinations that yield an RSS that exceeds the receiver's sensitivity by a margin $\epsilon.$ The path discovery process is demonstrated in Figure~\ref{fig:uncoordiated_beamforming}(a).  Both $A$ and $D$ sweep the plane but with different periods. $A$ sweeps slowly whereas $D$ sweeps fast. The protocol proceeds in rounds. During a round, $A$ fixes its antenna at one sector whereas $D$ sweeps through all sectors in random order. Consider $A$ pointing to $S_A$ whereas $D$ is pointing to $S_D$. Device $D$ sends a probe (sector sweep frame) to $A$. If the probe is correctly decoded and the RSS exceeds the receiver sensitivity by the margin $\epsilon$, $A$ replies with each own probe from sector $S_A.$ If $A$'s probe is also decodable, $D$ records sector $S_D$ as viable and responds to $A$ with an ACK. Upon reception of the ACK, $A$ also records $S_A$ as viable. This process is repeated for every sector of $D,$ while $A$ remains on the same sector. If multiple sectors of $D$ can be received by $A$, $D$ chooses the strongest one (highest RSS). In the next round, $A$ randomly switches to another sector and the probing process is repeated. 

At the end of the path discovery process, each party learns the set of viable sectors that can achieve communication with the other party. Note that each party only knows its own viable sectors. In the example of Figure~\ref{fig:uncoordiated_beamforming}(a), we show several rounds of path discovery. In round 1, $A$ steers to sector $S_A =1$ while $D$ sweeps through all eight sectors. Successful communication takes place in sector pairs $1-5$ and $1-6$. $A$ records sector 1 as a viable sector and $D$ records sectors 5 and 6 as viable sectors. Because the RSS on sector 5 is higher, only sector 5 is recorded. The process is repeated on future rounds. At the completion of 8 rounds, two paths yield viable directions, namely 1-5 (LoS) and 2-4 (reflection). 

{\bf 2. Uncoordinated path hopping.} In the next phase, $A$ and $D$ hop through their viable sectors in an uncoordinated fashion. The dwell time on each sector is sufficient to exchange several probes. If the probes are successfully exchanged, both devices record the CSI as part of the set used to generate each fingerprint. Otherwise, the CSI samples are discarded. In Figure~~\ref{fig:uncoordiated_beamforming}(b) we show $A$ and $D$ hopping in an uncoordinated fashion between their viable sectors discovered in Figure~~\ref{fig:uncoordiated_beamforming}(a). Only the blue pairs result in successful probe exchanges that are used for fingerprint generation.

\begin{figure}[t]
\begin{tabular}{c} 
\includegraphics[width=0.7\columnwidth]{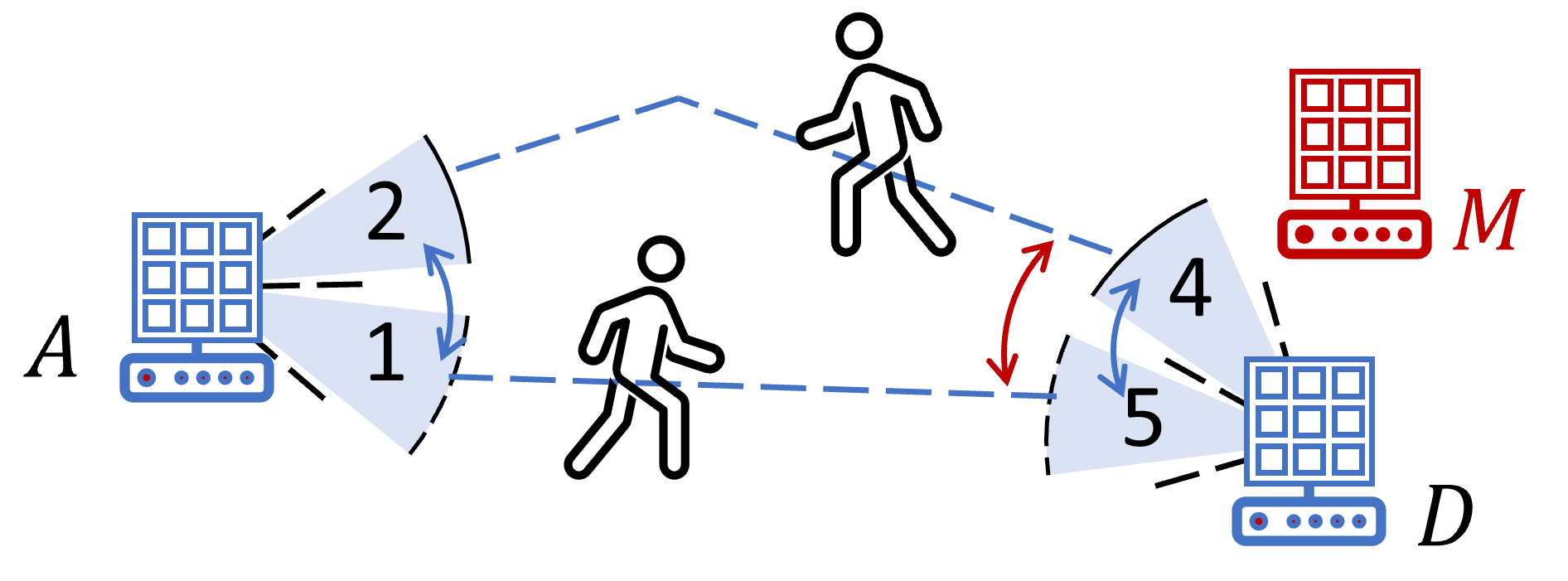} \\
(a) Passive adversary close to $D$.\\
\includegraphics[width=0.6\columnwidth]{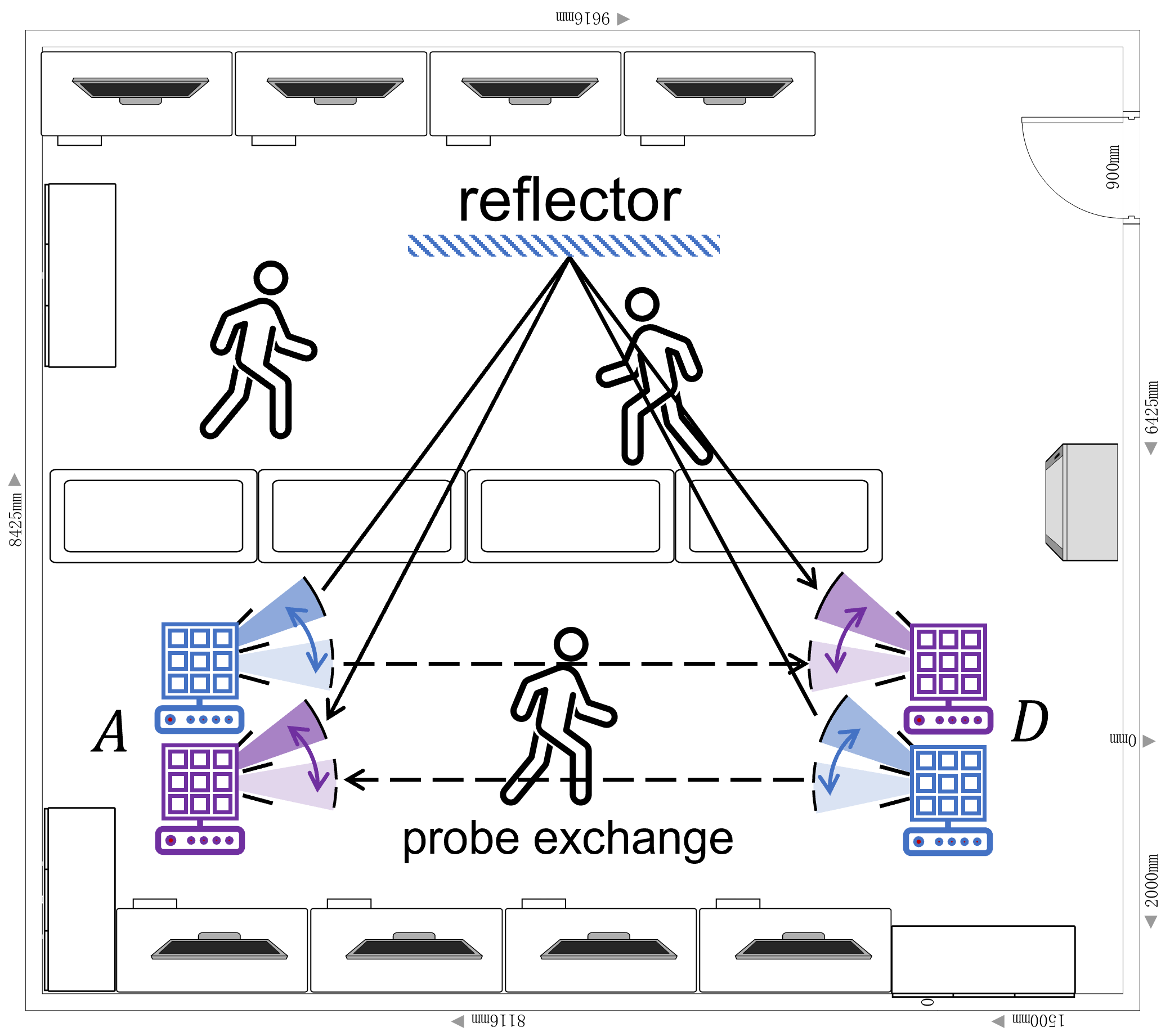}\\
(b) Setup for Jellybean+
\end{tabular}
\vspace{-0.1in}
\caption{ (a) passive threat model and (b) experiment setup.}
\label{fig:closebyM_UBM}
\vspace{-0.2in}
\end{figure}

\subsection{Security Analysis of Jellybean+} The security of Jellybean+ hinges on the unpredictability of the path that is sampled at any given time. Because the antennas are independently steered, Mallory does not know which paths are being sampled, unless she can capture all probes at all times.

{\bf Resistance to passive adversaries.}
We first consider the case where $M$ is positioned close to one of the devices, say $D$, such that it can observe similar viable paths as $D.$ This model is shown in Figure~\ref{fig:closebyM_UBM}(a). This adversary was shown to be capable of extracting a similar fingerprint as $D$ if only the LoS path was used. To evaluate Jellybean+, we assume the best case for $M$ where she knows all viable paths identified during the path discovery process. Referring to the example of Figure~\ref{fig:uncoordiated_beamforming}(a), this would correspond to knowing that paths 1-5 and 2-4 were discovered. This could be inferred by either knowing the device topology or intercepting probe messages. To measure the same fingerprint $F_D,$ Mallory must hop between the viable paths in an uncoordinated fashion, similar to $A$ and $D$. The probability of extracting the same fingerprint equals the probability of following the same path sequence as $D$ for the fraction of hops that correspond to viable paths (blue-shaded slots in Figure~\ref{fig:uncoordiated_beamforming} (b)). This probability is expressed in the following proposition. 

\begin{proposition}
\label{passive passing}
Let $A$ and $D$ hop between $P$ viable paths using uncoordinated sequences that match at $Q$ hops. The probability that $M$ extracts the same fingerprint as $D$ is given by 
\[
Pr[F_M=F_D] =\left(\frac{1}{P}\right)^Q.
\]
\end{proposition} 
\begin{proof}
To measure the same fingerprint $F_D$, $M$ must steer her antenna to the same sector as $D$ on all of the $Q$ hops where viable paths are observed. Each hop in the sequence is randomly selected out of $P$ paths with probability $1/P.$ Matching a single hop occurs with the same probability of selecting the specific path out of $P$ paths, which is equivalent to a Bernoulli trial with success probability $1/P$. Each hop in the hopping sequence is selected independently, making the probability of matching $X$ out of $Q$ hops follow the  binomial distribution $B(Q,1/P).$ Because we want all $Q$ hops to match it follows that   $Pr[F_M=F_D] = P(Q)= (1/P)^Q.$  
\end{proof}

{\bf Remark 1:} Proposition 1 holds even if Mallory is equipped with $P$ antennas that can sample physical activity an all $P$ paths simultaneously. This is because Mallory is not aware of where $D$ steers its antenna at every hop and therefore she has to guess which of the $P$ paths she should include in the CSI sample set. A correct guess occurs with probability $1/P$  independently at each viable hop, yielding Proposition 1. 

{\bf Remark 2:} When relaxing the fingerprint matching condition to only $K$ out of $Q$ paths, the probability of inferring a {\em similar} fingerprint to $F_D$ is given by the binomial distribution $B(Q,1/P)$, with $K$ out of $Q$ successes. 

\begin{figure}[t]
\centering
\begin{tabular}{c} 
\includegraphics[width=0.8\columnwidth]{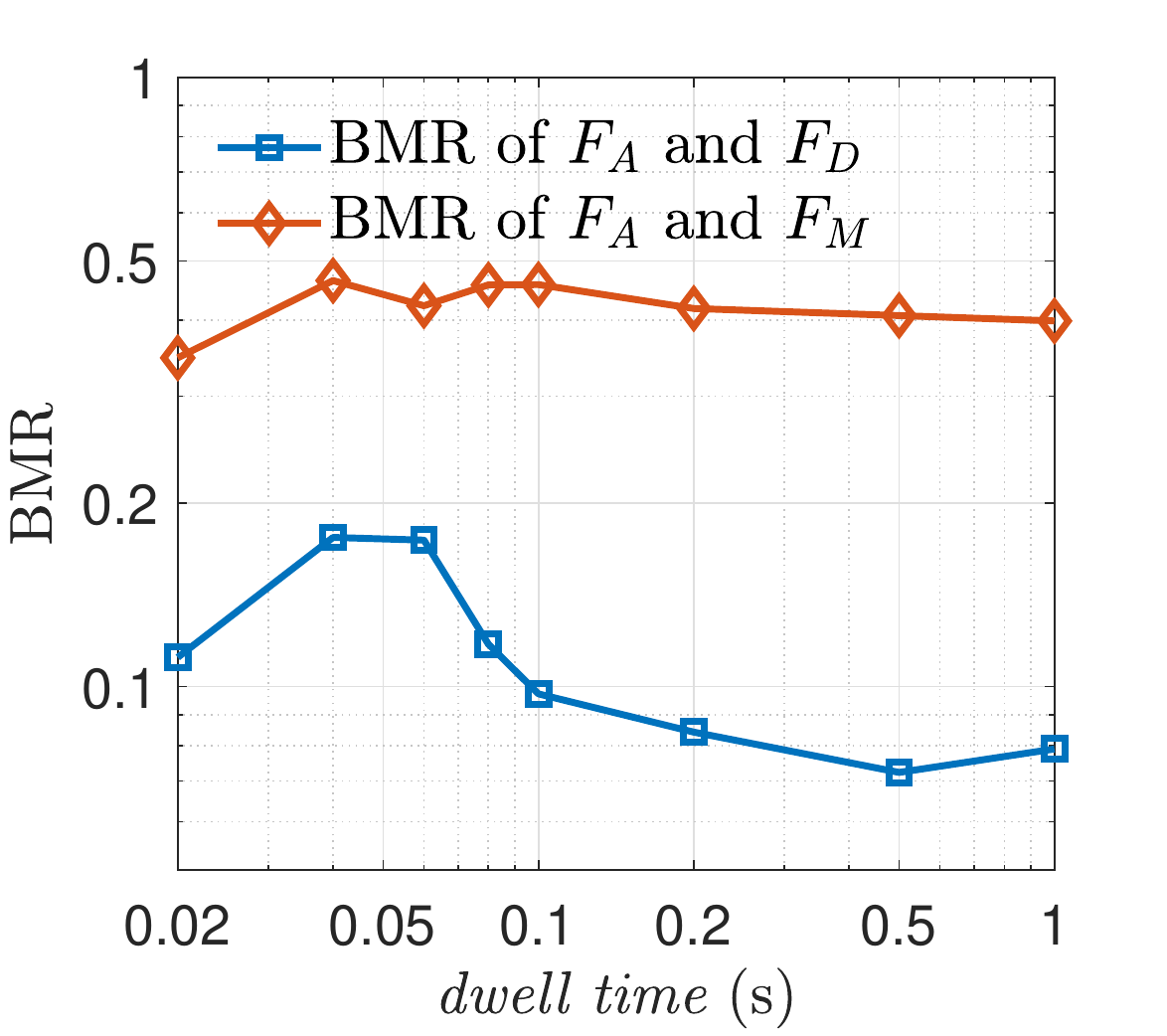} 
\end{tabular}
\vspace{-0.1in}
\caption{Security evaluation of the Jellybean+ when $M$ is co-located with $D$, where BMR between $A-D$ and $A-M$ as a function of the path dwell time.}
\label{fig:UBM}
\vspace{-0.15in}
\end{figure}

\begin{figure*}[t]
\centering
\setlength{\tabcolsep}{-7pt}
\begin{tabular}{ccc}  
\includegraphics[width=0.78\columnwidth]{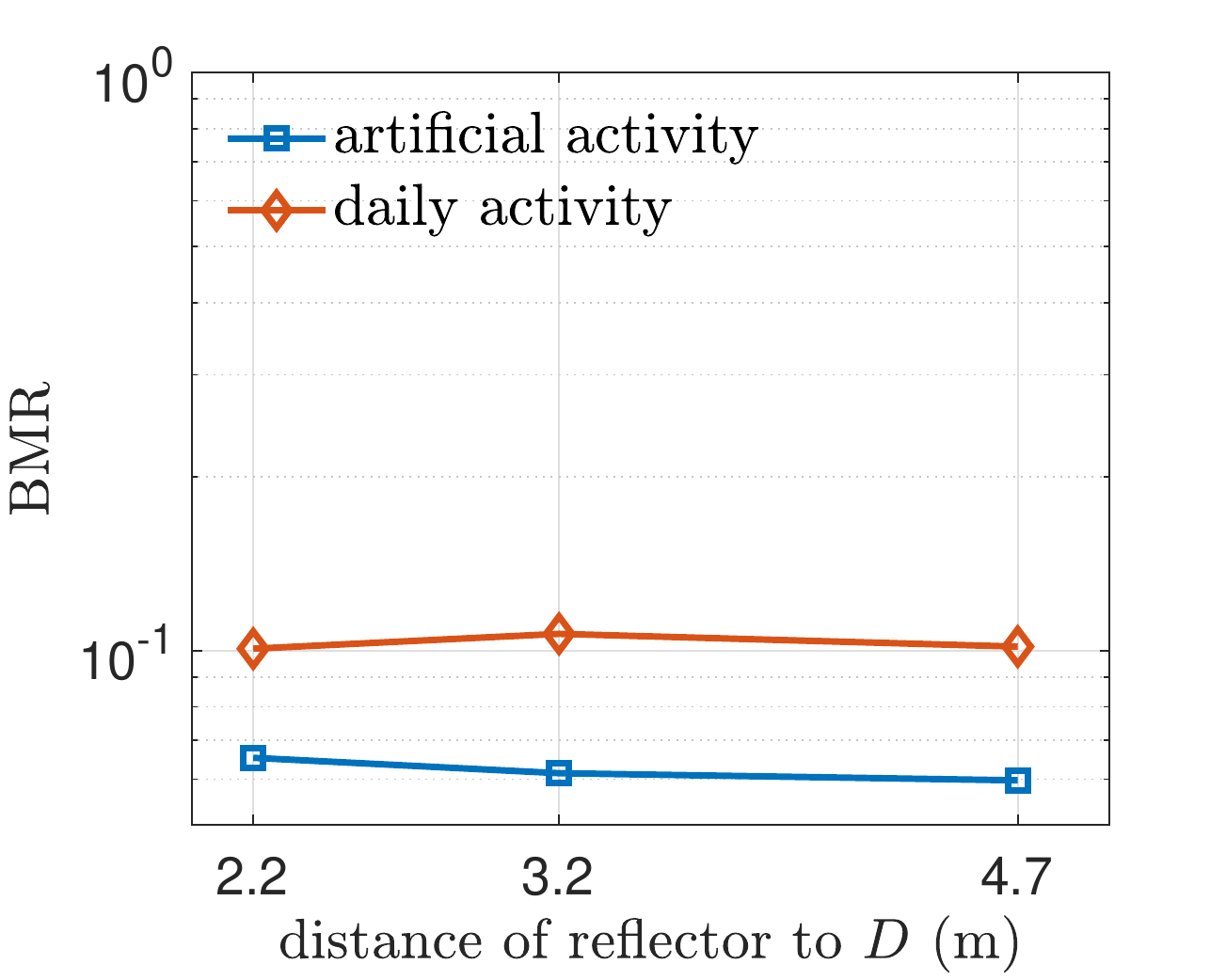} & 
\includegraphics[width=0.78\columnwidth]{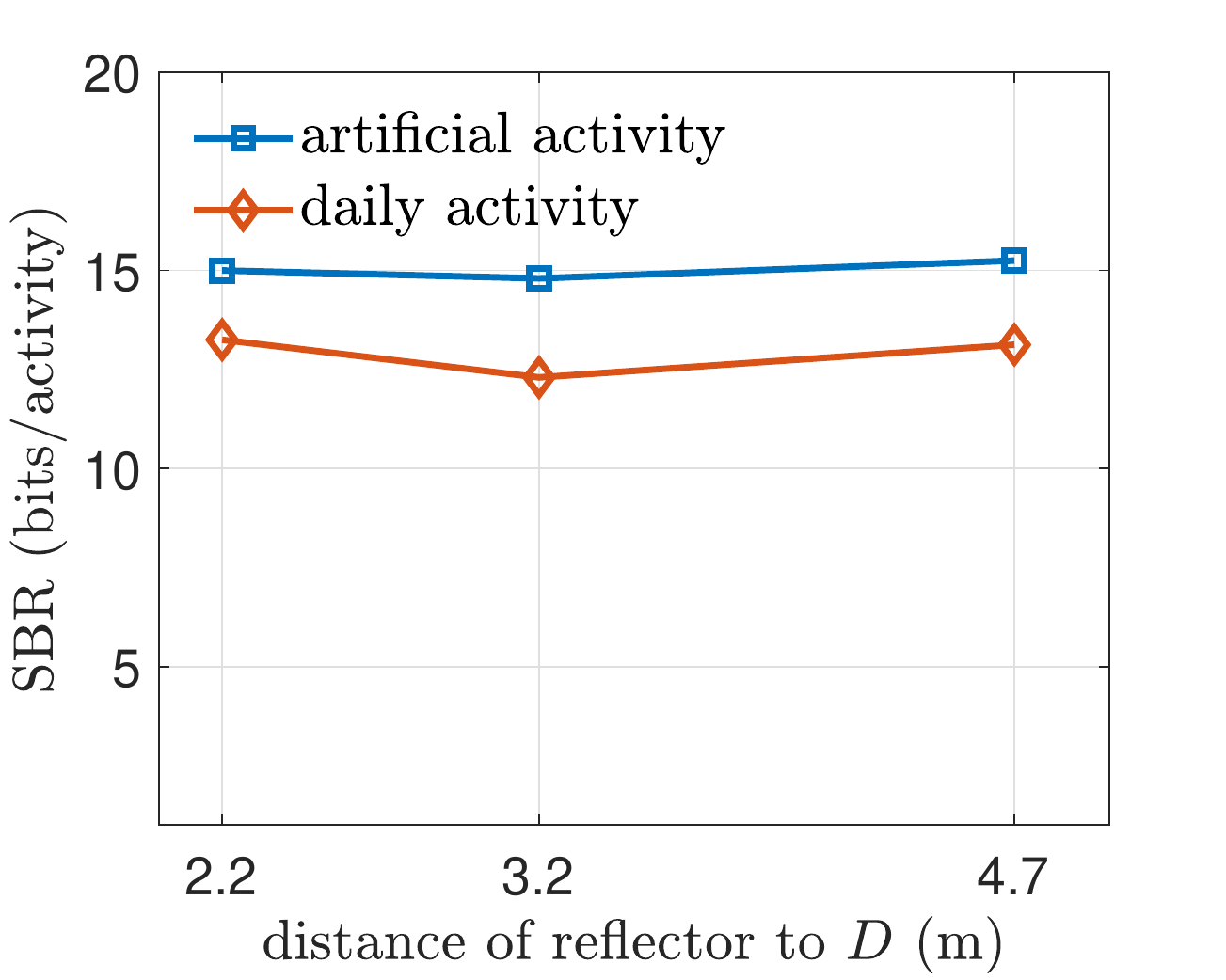} &
\includegraphics[width=0.78\columnwidth]{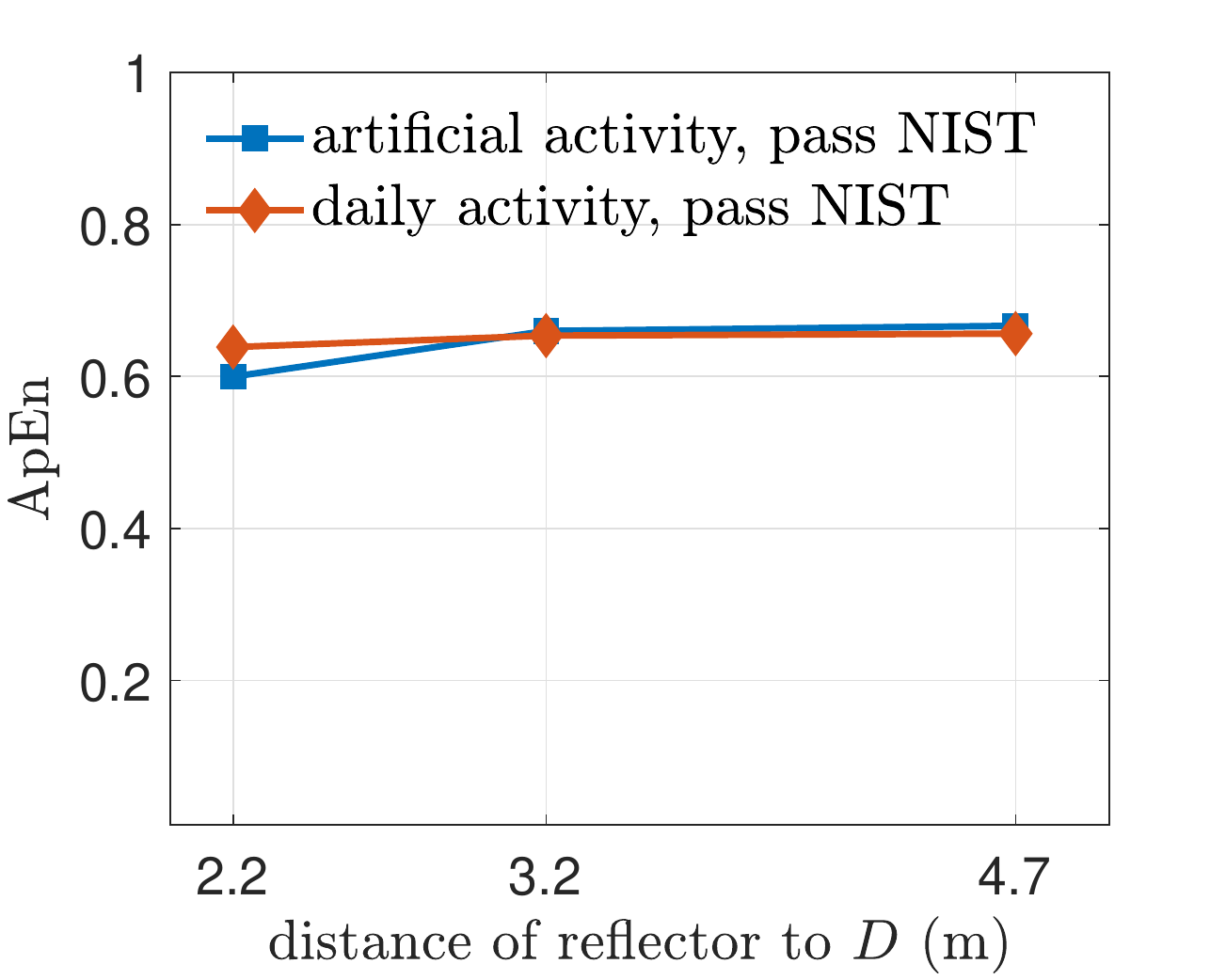}
\end{tabular}
\caption{Performance evaluation of the Jellybean+ with two types of activities in experimental room B, where the evaluation metrics BMR, SBR, and ApEn are captured in a function of the distances of $D$ to the reflector on the $A-D$ NLoS path.}
\label{fig:refl_path}
\vspace{-0.15in}
\end{figure*}

Note that we cannot analytically derive the probability of inferring the key $k$ (equivalently the BMR between $F_D$ and $F_M$) because it depends on the start time and duration of the observed activities on each path. We experimentally evaluated the BMR of the adversary under the UPH protocol, using the setup shown in Figure~\ref{fig:closebyM_UBM}(b). During the path discovery phase, $A$ and $D$ discovered two paths, one through the LoS and one through the reflector. We collected CSI samples on each of those paths for 25 minutes independently. We generated hopping sequences at $A$, $D$, and $M$ independently, and varied the {\em dwell time} on each path from 0.02 seconds to 1 second. We then followed the uncoordinated path hopping mechanism between $A$ and $D$ according to the generated hopping sequences and recorded the BMR. The adversary recorded CSI samples, only when it could hear probes from $A$. 

 Figure~\ref{fig:UBM} shows the BMR as a function of the dwell time on each path. We observe that the BMR between $F_A$ and $F_D$ remain unaffected by the introduction of UPH. On the other hand, the BMR between $F_A$ and $F_M$ is approximately 0.5 which equals the probability of guessing the correct path at every hop. The success  probability is also equal to the probability of randomly guessing $k$,  indicating that the adversary gains no advantage by attempting to observe activities from nearby when path diversity is exploited. 
 
 Note that the BMR remains relatively constant for all dwell times. This is because $M$ hopped paths with the same period as $D.$ Hopping between paths at a shorter period is not beneficial to the adversary because (a) it will miss receiving probes that are essential for determining the activity start time and duration and (b) it still has to guess which paths contribute to the fingerprints (Remark 1). 

{\bf Resistance to active adversaries.} We now describe how Jellybean+ can prevent active attacks. An active attack relies on beam-stealing to place Mallory in control of the probe exchange between $A$ and $D$, so that $M$ can independently establish keys with $A$ and $D$, using parallel pairing sessions (MitM). The beam-stealing attack is possible, because the beam alignment protocol results in the use of a single path (the best path) between two communicating devices. With UPH, multiple paths are used to extract activity fingerprints. Even if the adversary controls one or a subset of those paths, she still needs to guess when $A$ and $D$ sample the remaining attack-free paths.  Moreover, $M$ cannot sample other paths she does not control unless she simultaneously deploys passive eavesdroppers that are aligned to those paths. In this case, the results obtained for passive adversaries (Proposition 1 and Figure~\ref{fig:UBM}) hold for the active adversary. As long as there are more than one paths not controlled by $M$, the probability of matching fingerprints decays exponentially with the hopping sequence length. 

\subsection{Performance Evaluation of Jellybean+}

{\bf Time analysis.} One concern with UPH is that the uncoordinated hopping will introduce a long delay until devices can measure sufficiently long fingerprints. To alleviate this concern, we analyze the time for fingerprint generation. The path discovery phase is deterministic and requires the operation of $N$ full sweeps from $D$ where $N$ is the number of available sector ($N^2$ hops in total). In each hop, $A$ and $D$ exchange two probes (sector sweep frame) and a sweep ACK. Using parameters from the IEEE 802.11ad specification \cite{IEEE:802.11ad}, each probe is 26 bytes long whereas a sweep ACK is 28 bytes long. Moreover, the time interval between two probes for beam switching, called the short beamforming interframe spacing (SBIFS) is $T_s=1 \mu s$. Assuming a 160Mbps transmission rate, the time consumed during path discovery is  $5.12 ms$.

The delay for fingerprint extraction during the UPH phase is similar to that of Jellybean if the dwell time is selected to be low (e.g., less than 0.1sec). This is because path hopping is much faster than the activity timescale, so that all activities occurring on any of the viable paths will be intermittently sampled. 

{\bf Performance on the reflection path.} Since Jellybean+ uses a reflection path, we explored the  the BMR, SBR, and ApEn on that path. We varied the distance between $D$ and the reflector as depicted in room B of Figure~\ref{fig:exp_set}(c). The experiments were executed for 3 minutes under each distance with artificial and natural daily activities. Figure~\ref{fig:refl_path} shows the BMR, SBR, and APEn for the two types of activities. We observe that the BMR, SBR, and ApEn have similar values to those obtained in the LoS path for each activity type, despite the lower RSS on the reflected path. The use of CSI variance is a strong indicator of activity and inactivity. In fact, the reflected path is easier to disrupt due to its higher attenuation.

\section{Related Work}

{\bf Context-based Device Pairing.} Our proposed protocol belongs to  the general category of context-based device pairing protocols. Such protocols aim to establish a shared secret key between two wireless devices by sensing common contexts that are dynamic and random, without relying on any  pre-shared secrets such as   passcodes. They provide notable advantages such as reduced user efforts, perpetual key renewal, and a high level of security against remote adversaries. Commonly explored physical modalities in context-based pairing include ambient light \cite{miettinen2014context}, sound \cite{xu2021inaudiblekey, miettinen2014context, schurmann2011secure}, human operation \cite{li2020t2pair, zhang2018tap}. However, all the above methods require that all the devices possess the same sensing modality, which can be impractical for real-world applications such as IoT. Recently, several works addressed this limitation by using  multi-modal sensing  \cite{han2018you, farrukh2022one} for common activity detection, which is applicable to  devices with heterogeneous sensing capabilities. However, the types of activities that can be simultaneously sensed by different sensors are limited.  Most importantly, all those methods are only secure against a passive adversary which is located outside of the physical context (e.g.,  a   room where the legitimate devices are located). 

On the other hand, other works such as \cite{liu2013fast,mathur2008radio, mathur2011proximate, jana2009effectiveness, premnath2012secret, liu2012collaborative, zeng2010exploiting} used  the (in-band)  radio interface to sense ambient RF  signals do not require devices to possess a common sensor other than the same RF interface which is widely existed. They have leveraged the correlation of measured ambient RF signals (e.g., RSS) to securely pair two or more devices that are co-located or in   proximity \cite{mathur2011proximate, shi2012bana, wu2018survey}. However,  RF signals decorrelate fast when the distance increases (usually not correlated beyond  half wavelength), which requires  very small proximity ranges (e.g., 12.5 cm at 2.4GHz, and millimeter range for mmWave frequencies) that are impractical for many applications.  In \cite{Xu_2022}, large-scale fading of RF signals is leveraged for verification of continuous proximity, which has a larger correlation distance, but it is applicable to outdoor mobile settings and not in the mmWave band.  
In contrast, our protocol does not limit the distance between the legitimate devices as long as they are within each other's communication range.   Also, it  is secure against an adversary  in the same environment (even in close proximity to legitimate devices).

{\bf Security of mmWave Communications.} 
Early works on the security of mmWave focused on ensuring message confidentiality against passive adversaries. mmWave is generally considered to be more secure against eavesdropping due to its highly directional transmissions. Indeed, this is the case when the adversary is located outside of the  beam of a mmWave link. Yet, Steinmetzer et al. \cite{steinmetzer2015eavesdropping} demonstrated that  eavesdropping is still possible if the adversary can place a reflector within the beam to deflect the signal. There are also works that leverage channel reciprocity to extract secret keys from a mmWave link \cite{jiao2018physical,li2019physical, jiao2019physical}, however, it is challenging since the entropy is  low under static channel environments, while channel reciprocity may not hold for dynamic channels due to the high sensitivity of mmWave signals. 

On the other hand, several works studied the security of mmWave communication against active attacks such as impersonation or MitM. Balakrishnan et al. \cite{balakrishnan2019physical} studied device fingerprinting based on variations in their beam patterns, in order to prevent impersonation attacks. Wang et al. \cite{wang2021exploiting, wang2020machine} adopted a similar approach  by using machine learning algorithms to extract unique beam patterns for each device, enabling the receiver to distinguish  devices and detect spoofing attacks. However,   these methods require training, which incur extra overhead and complexity, and it is impractical to obtain secure training data when new devices join the network. Steinmetzer et al. \cite{steinmetzer2018authenticating}  first demonstrated the beam-stealing attack against the 802.11ad protocol for mmWave communications, and proposed an authenticated beam sweep protocol that integrates crypto machinery to protect it against beam-stealing. However, in the device pairing problem, there is no prior trust which makes such an approach inapplicable. 

{\bf Wireless Sensing for Human Activity Detection and Recognition (HAR).} 
Our work is also related to wireless sensing, especially passive sensing where the wireless transceivers are deployed in the environment rather than on the objects of interest. Existing works in this domain typically extract features from the CSI of the received wireless signals to detect or recognize human activity, based on the fact that different activities will incur different disturbances to the CSI.  For example,  in sub-6GHz bands, WiFi signals have been widely used for activity sensing \cite{venkatnarayan2018multi, gao2021towards, tan2016wifinger, wang2014eyes, wang2017device}. Meanwhile, due to the high bandwidth and sensing resolution of mmWave signals, mmWave sensing has attracted increasing attention in recent years \cite{liu2020real, singh2019radhar, wang2021m, sun2019privacy, jiang2018towards, vandersmissen2020indoor, lien2016soli, wang2016interacting}. 

Our work differs from HAR,  since our activity fingerprinting method only aims at detecting the presence or absence (timing) of an activity, rather than  classifying the type of activity.  In contrast, human activity recognition typically relies on machine learning methods  which require a large amount of training data. It is impractical to collect sufficient and secure training data in device pairing scenarios (where the devices meet for the first time without prior trust). In our protocol, we only require that the devices can observe common events' timing, thus training is not needed.

\section{Conclusion}
We proposed Jellybean, a novel secret-free  pairing protocol for devices operating in the mmWave band. Our basic idea is to leverage in-band mmWave communications to detect (human) activities present in the physical environment as the common  context.  Jellybean extracts activity timing information   via CSI measurements, which is used as a  source of shared entropy for  key agreement. We further enhance the security of the basic scheme with an uncoordinated path hopping mechanism, which achieves activity detection using directional beams under the absence of prior trust. The key novelty of Jellybean originates from exploiting the unique properties of mmWave signals to defend against both co-located passive adversaries  as well as active adversaries. We evaluated Jellybean experimentally in two indoor environments, and showed that it can securely pair devices  in the same room while thwarting  adversaries located in the same space. 

\bibliographystyle{ACM-Reference-Format}
\bibliography{cite}

\end{document}